\documentclass[psamsfonts]{article}
\usepackage{amssymb,amsmath, amsthm, setspace, color, graphicx, fullpage}
\usepackage{mathpazo, mathptmx, flexisym}
\usepackage{fancyhdr}
\usepackage[authoryear]{natbib}
\usepackage{url}
\usepackage{breqn}
\usepackage{functan}
\usepackage[margin=1in]{geometry} 
\usepackage{setspace}             
\usepackage[all, cmtip]{xy}
\usepackage[english]{babel}
\usepackage{graphicx}             
\usepackage{float}
\usepackage[section]{placeins}
\usepackage{mathtools}
\usepackage{etoolbox}

\usepackage[pdfborder={0 0 0},final=true,colorlinks=true,linkcolor=magenta,citecolor=blue]{hyperref}
\usepackage[capitalize,nameinlink,noabbrev]{cleveref}
\numberwithin{equation}{section}
\DeclareMathSizes{12}{14}{14}{10}             
\newcommand{\tab}{\hspace*{2em}}              

\newcommand{\bm}[1]{\mathbf{#1}}                     
\newcommand{\diadt}[2]{t^{(#1)}_{#2}}                

\Macro{xj}{\bm{x}_{\diadt{n}{j}}}                              
\Macro{xjplus1}{\bm{x}_{\diadt{n}{{j+1}}}}                              
\Macro{xTj}{\bm{x}^T_{\diadt{n}{j}}}                           
\Macro{xTjplus1}{\bm{x}^T_{\diadt{n}{{j+1}}}}                           
\Macro{zj}{\bm{z}_{\diadt{n}{j}}}                              
\Macro{zjplus1}{\bm{z}_{\diadt{n}{{j+1}}}}                              
\Macro{zTj}{\bm{z}^T_{\diadt{n}{j}}}                           
\Macro{zTjplus1}{\bm{z}^T_{\diadt{n}{{j+1}}}}                           
\Macro{Mxj}{M_{\bm{x}_{\diadt{n}{j}}}}                         
\Macro{Mxjplus1}{M_{\bm{x}_{\diadt{n}{{j+1}}}}}                         
\Macro{Mzj}{M_{\bm{z}_{\diadt{n}{j}}}}                         
\Macro{Mzjplus1}{M_{\bm{z}_{\diadt{n}{{j+1}}}}}
\Macro{SxO}{S\times\Omega}                           
\Macro{SigmaSxO}{\mathfrak{S}\times\mathfrak{O}}     
\Macro{Lp}{L^p_{}(X,\mathfrak{B}(X),m)}              
\Macro{L2}{L^2_{}(X,\mathfrak{B}(X),m)}
\Macro{LinftyS}{L^{\infty}_{}(S \times \Omega,\mathfrak{S}\times\mathfrak{O},\mu)}
\Macro{LinftyS*}{L^{\infty}_{}(S \times \Omega,\mathfrak{S}\times\mathfrak{O},\mu)^*}
\Macro{Mvof}{\mathfrak{M}_{v \circ f}^{}}              
\Macro{ImPhi}{\text{Im}\;\phi}                         
\Macro{KerPhi}{\text{Ker}\;\phi}                       
\Macro{GQKerPhi}{G\backslash\text{Ker}\phi}            
\Macro{CoMap}{\alpha\circ\delta\circ\gamma}             
\Macro{GrpA}{(G,*)}                                    
\Macro{GrpO}{(G^\prime,\diamond)}                      
\theoremstyle{plain}
\newtheorem{thm}{Theorem}[section]

\newtheorem{prop}[thm]{Proposition}
\theoremstyle{plain}
\newtheorem{assumption}[thm]{Assumption}      
\theoremstyle{definition}

\theoremstyle{definition}

\theoremstyle{remark}

\crefname{thm}{theorem}{theorems}
\Crefname{thm}{Theorem}{Theorems}
\crefname{cor}{corollary}{corollaries}
\Crefname{cor}{Corollary}{Corollaries}
\crefname{lem}{lemma}{lemmas}
\Crefname{lem}{Lemma}{Lemmas}
\crefname{prop}{proposition}{propositions}
\Crefname{prop}{Proposition}{Propositions}
\crefname{assumption}{assumption}{assumptions}
\Crefname{assumption}{Assumption}{Assumptions}
\crefname{defn}{definition}{definitions}
\Crefname{defn}{Definition}{Definitions}
\crefname{rem}{remark}{remarks}
\Crefname{rem}{Remark}{Remarks}

\newif\ifinappendix
\AtBeginEnvironment{thm}{\singlespacing}
\AtEndEnvironment{thm}{\ifinappendix\singlespacing\else\onehalfspacing\fi}
\AtBeginEnvironment{cor}{\singlespacing}
\AtEndEnvironment{cor}{\ifinappendix\singlespacing\else\onehalfspacing\fi}
\AtBeginEnvironment{lem}{\singlespacing}
\AtEndEnvironment{lem}{\ifinappendix\singlespacing\else\onehalfspacing\fi}
\AtBeginEnvironment{prop}{\singlespacing}
\AtEndEnvironment{prop}{\ifinappendix\singlespacing\else\onehalfspacing\fi}
\AtBeginEnvironment{assumption}{\singlespacing}
\AtEndEnvironment{assumption}{\ifinappendix\singlespacing\else\onehalfspacing\fi}
\AtBeginEnvironment{assumptionletter}{\singlespacing}
\AtEndEnvironment{assumptionletter}{\ifinappendix\singlespacing\else\onehalfspacing\fi}
\AtBeginEnvironment{defn}{\singlespacing}
\AtEndEnvironment{defn}{\ifinappendix\singlespacing\else\onehalfspacing\fi}
\AtBeginEnvironment{rem}{\singlespacing}
\AtEndEnvironment{rem}{\ifinappendix\singlespacing\else\onehalfspacing\fi}
\AtBeginEnvironment{proof}{\singlespacing}
\AtEndEnvironment{proof}{\ifinappendix\singlespacing\else\onehalfspacing\fi}
\AtBeginEnvironment{userdefined}{\singlespacing}
\AtEndEnvironment{userdefined}{\ifinappendix\singlespacing\else\onehalfspacing\fi}

\begin{document}\singlespacing
\title{\textbf{Equivalent Behavioral Martingale Measure}}
\author{G. Charles-Cadogan \thanks{School of Accounting and Finance, University of Leicester; and Institute for Innovation and Technology Management, Ted Rogers School of Management, Toronto Metropolitan University. e-mail: \textcolor[rgb]{0.00,0.00,1.00}{\href{mailto:gocadog@gmail.com}{gocadog@gmail.com}}; Tel: (786) 329-5469.} \\~\\Working Paper\\Comments welcome}
\date{\today}
\maketitle
\begin{abstract}
\noindent The paper places the Samuelson--Merton ``util-prob'' construction and preference-based state-price interpretation inside a continuous-time equivalent-martingale-measure framework. The contribution is not the existence of a marginal-utility pricing kernel, but a compensated stochastic-exponential representation of the ``util-prob'' density. Behavioral weighting by marginal utility becomes an equivalent martingale measure only after the drift generated by risk aversion and prudence is compensated. Under positivity, integrability, normalization, and martingale conditions, the behavioral density is generated by the Brownian kernel $-A_t\sigma_t$, where $A_t$ is Arrow--Pratt risk attitude and $\sigma_t$ is state volatility. The required compensator identifies the finite-variation term that must be removed before marginal-utility weighting defines a true martingale density. The framework preserves put--call parity under a common admissible kernel and yields a Cox--Ross risk-premium identity linking option and underlying premia through volatility exposures. Empirically, the paper maps option-implied and physical densities into a signed Arrow--Pratt decomposition, with positive curvature interpreted as risk aversion and negative curvature as risk-seeking or speculative tilt. In a panel of liquid index ETF and single-name options, arbitrage-penalized EBMM has lower pooled holdout squared-error loss than arbitrage-penalized SVI when both models receive the same total-variance positivity and butterfly-density penalties, although slice-level results remain mixed. The exercise is best read as a density-ratio admissibility test rather than a replacement for production volatility-surface estimation.
\end{abstract}

\tableofcontents
\onehalfspacing
\section{\sc{Introduction}}\label{sec:introduction}
In the textbook Black--Scholes--Merton formula for a non-dividend-paying European option with known rate, spot, strike, and maturity, volatility is the remaining input inferred from option prices.\footnote{The Black--Scholes--Merton options pricing formula for a call option $C(S,K,r,T-t)$ where $S$ is stock price, $K$ is strike price, $r$ is risk-free rate, and $T-t$ is time to maturity, is given by
\begin{equation}\label{BSM}
\begin{split}
C(S,K,r,T-t) = S\Phi(d_1)-Ke^{(-r(T-t) )} \Phi(d_2) \\
d_1=\frac{\log(S/K)+(r + \frac{1}{2}\, \sigma^2)(T-t)}{\sigma\sqrt{T-t}} \quad
d_2=\frac{\log(S/K)+(r - \frac{1}{2}\, \sigma^2)(T-t)}{\sigma\sqrt{T-t}}
\end{split}
\end{equation}
where $\Phi$ is the standard normal cumulative distribution function and $\sigma$ is the volatility of the underlying.} Taking the Black--Scholes--Merton formula as a benchmark for plain-vanilla option pricing, implied volatility is a market-implied risk-neutral volatility parameter, often interpreted as a forward-looking summary of volatility compensation and state-price weighting. Market commentary often treats implied volatility as a gauge of investor sentiment or risk appetite, but the interpretation used here is narrower: option prices identify a pricing-kernel weighting over states.

\tab One of the empirical regularities of implied volatility is a U-shaped pattern or skew which takes the form of a smile or smirk. See e.g., \citet[Chapter 16]{Hull2006}; \citet[Chapter 4]{Rebonat1999}. Because this pattern was not observed before the crash of 1987, some analysts believe that it reflects ``crash-o-phobia''
\footnote{See e.g., \citet[p. 775]{Ruben1994} (``The market's pricing of index options since the crash seems to indicate an increasing crash-o-phobia, a phenomenon that we will subsequently document in other ways.'')}
, i.e. it is a behavioral response to market crash.

\tab This paper does not claim that marginal-utility pricing kernels are new. Its contribution is to show when the Samuelson--Merton ``util-prob'' density can be represented as a true equivalent martingale density. The required object is not raw marginal utility, but a compensated stochastic exponential whose Brownian kernel is $-A_t\sigma_t$, where $A_t$ is the Arrow--Pratt risk-attitude process and $\sigma_t$ is state volatility. This representation gives a martingale-admissibility filter for behavioral interpretations of option-implied smiles and smirks.

\tab The goal of this paper is to provide a structural interpretation of volatility smiles and smirks by introducing an equivalent behavioral martingale measure (``EBMM"), and the behavioral restrictions required for it to hold. The term ``behavioral'' is used in the preference-based sense of state-dependent marginal-utility weighting; it does not imply a violation of no-arbitrage or a rank-dependent probability distortion. Thus EBMM is behavioral only in the sense of heterogeneous or state-dependent marginal utility; it remains fully compatible with no-arbitrage pricing. \citet[pg. ~161]{Merton1973a} laid a foundation for this approach when he stated (parenthetically) that the Black-Scholes-Merton options pricing formula is functionally equivalent to the earlier \citet{SamuelsonMerton1969} utility dependent behavioral options pricing formula--if the ``util-prob'' density $dQ(z)$
\footnote{\citet[p.~8, note~2]{Marsh1984} summarizes the Samuelson--Merton ``util-prob'' density as
\[
  dQ(z)=\frac{U^\prime(z^*)\,dP(z)}{E_P[U^\prime(z^*)]},
\]
where $z^*$ denotes the return generated by the expected-utility-maximizing portfolio. Thus $Q$ is the physical probability measure reweighted by normalized marginal utility.}
in the following options price formula is replaced by a lognormal density.
\begin{equation}\label{SamuelsonMerton1969}
C(S,K,r,T-t)=e^{-r(T-t)} \int_{\frac{K}{S}}^{\infty}(ZS-K)dQ(Z;T-t)
\end{equation}

\tab Here $dQ$ is a risk adjusted ``util-prob'' distribution dependent on both risk preferences and aggregate supplies, while the distribution in \cite{BS1973} is the objective distribution of returns on the common stock. In fact, \citet{SamuelsonMerton1969} ``util-prob'' formula ``will be the same [as Black--Scholes] only in the special case when $dQ$ is a lognormal density with the variance of $\log(Z)$ equal to $\sigma^2 [(T-t)]$," Merton, \textit{supra}. In particular, \citet{SamuelsonMerton1969} plainly state `The probability-cum-utility function Q(z;n) has all the properties of a probability distribution, but it weights the probability of each outcome so to speak by the marginal utility of wealth in that outcome,''  \citet[pg. ~168]{Merton1992}. Thus $Q(z;n)$ is a marginal-utility-weighted probability measure. It is related to, but conceptually distinct from, the rank-dependent probability weighting scheme in \citet{Quiggin1982,Quiggin1993}. Empirical research has shown that returns are not lognormally distributed
\footnote{The thick-tail property of stock returns is well documented in the empirical finance literature.}.
 So, \citet{SamuelsonMerton1969} provides a foundation for behavioral options pricing 
by virtue of utilizing subjective probability measures constructed from utility metrics.\footnote{  \cite{KahnTver1979, TverKahn1992} also introduce a probability weight function to account for gains and losses relative to a reference point. Their scheme involves employing \citet{Quiggin1982} approach by ranking outcomes and extracting comonotonic probability weights for risky prospects.} Evidently, the objectivist versus subjectivist approach to risk and uncertainty extends to options pricing. We extend \citet{SamuelsonMerton1969} to equivalent martingale measures in order to see how preference parameters interact in the Black-Scholes-Merton formula by and through volatility. This is tantamount to recovering risk aversion from a state-price density.\footnote{See \citet[pg. ~627]{BreLitzen1978} and \citet{Jac2000}.}

\tab The paper makes four more limited and testable claims. First, the Samuelson--Merton ``util-prob'' density is an equivalent measure only when marginal utility is positive, integrable, and normalized. Second, after It\^o correction, the compensated marginal-utility density is a stochastic exponential generated by $-\lambda_t$, where $\lambda_t=A_t\sigma_t$ is the Brownian risk price. Third, if an option and its underlying are priced by the same state-price kernel, the Cox--Ross risk-premium identity follows directly from equality of Brownian risk prices. Fourth, volatility smiles and smirks can be represented as asymmetric mixtures of positive and negative components of the Arrow--Pratt process across moneyness states. These claims place the paper close to stochastic discount factors, marginal-utility pricing kernels, equivalent martingale measures, utility-based pricing, and risk-aversion recovery from option prices. Its narrower contribution is to express the Samuelson--Merton ``util-prob'' density as a compensated stochastic exponential and to use the resulting martingale restrictions as an admissibility filter for behavioral interpretations of option-implied smiles.

\subsection{Contribution and relation to prior work}\label{subsec:contribution_prior_work}
\tab The paper is not meant to replace empirical risk-aversion recovery from option prices. It formalizes a bridge between that literature and the older ``util-prob'' construction. In the language of state-price densities, option data identify a pricing kernel or risk-neutral density relative to a physical density. This is the state-price recovery logic in \citet{BreLitzen1978}, \citet{JacRub1996}, and \citet{Jac2000}. Jackwerth-style recovery infers risk aversion from state-price and physical densities under utility restrictions; nonparametric and semiparametric state-price estimators fit the density directly \citep{AitSahaliaLo2000,FrassoEilers2022}. Recent work continues this recovery agenda in commodity, crypto, horizon-dependent, and volatility-risk-premium settings \citep{ChristoffersenJacobsPan2022,Cheng2026,AndriesEisenbachSchmalz2024,Duarte2024}. In the language used here, the same object can be represented as a compensated marginal-utility density. The novelty is therefore mathematical and interpretive: the paper shows exactly which drift compensation and equivalence conditions are needed before marginal utility can be treated as an equivalent martingale measure. Relative to non-behavioral smile parametrizations, the EBMM does not claim that hyperbolic risk-attitude components are the unique way to fit an implied-volatility surface. Its narrower contribution is to attach the curvature and skew of such a surface to an admissible marginal-utility density, so that recovered smiles, smirks, and heterogeneous-agent mixtures can be interpreted inside one density-based framework.

\tab The state-preference language is related to the higher-moment preference tradition in \citet{KrausLitzenberger1976}, while the density-recovery implementation follows the option-implied state-price logic of \citet{BreLitzen1978}.

\tab This also clarifies the empirical comparison. The option-implied risk-aversion literature typically recovers a state-dependent pricing kernel or implied risk-aversion function from the ratio of risk-neutral to physical densities \citep{AitSahaliaLo2000,Jac2000,RosEngle2000}. EBMM uses the same recoverability logic, but maps the recovered density ratio into an explicit stochastic-exponential kernel with interpretable Arrow--Pratt components. In the present hyperbolic specification, $\widehat A_{\mathrm{averse}}>0$ is the risk-averse curvature component, $\widehat A_{\mathrm{seek}}<0$ is the signed risk-seeking or speculative component, and zero is the risk-neutral boundary. Thus the empirical advantage is not merely fitting an option smile; it is extracting a signed risk-attitude decomposition subject to martingale admissibility.

\tab The paper's contribution is therefore threefold. First, it gives explicit sufficient conditions under which the Samuelson--Merton ``util-prob'' density becomes an equivalent martingale density. Second, it derives the compensator that converts marginal utility into a Dol\'eans--Dade stochastic exponential and identifies the corresponding Brownian risk price as Arrow--Pratt risk attitude times state volatility. Third, it uses this representation as an empirical admissibility filter: an option-implied smile or smirk receives a behavioral interpretation only if the recovered density ratio can be represented by a positive, normalized, true martingale density.

\tab Ross recovery provides a complementary identification route. Under Markovian and transition-independence restrictions, \citet{Ross2015} shows how state prices can recover the physical transition law and the pricing kernel from option prices. EBMM asks a different question. Conditional on a physical density, whether estimated separately or recovered by a Ross-style procedure, EBMM tests whether the pricing-kernel ratio can be generated by a compensated marginal-utility process whose Brownian risk price is $\lambda_t=A_t\sigma_t$ and whose density is a true stochastic-exponential martingale. Thus Ross recovery identifies a transition-based pricing kernel under Markov restrictions; EBMM identifies the utility-curvature representation, compensator, sign decomposition, and martingale-admissibility restrictions of a candidate kernel without requiring the Markov transition structure to be the source of identification.

\tab The rest of the paper proceeds as follows. \Cref{sec:Model} sets up the Samuelson--Merton ``util-prob'' construction and the Arrow--Pratt risk-attitude process. \Cref{sec:EBMM} derives the compensated equivalent behavioral martingale measure, states the marginal-utility restrictions, and extends the construction to heterogeneous-agent mixtures. \Cref{sec:RiskManagement} derives the Cox--Ross risk-premium identity, develops the smile/smirk representation, and reports the option-panel calibration, admissibility diagnostics, and predictive check. \Cref{sec:Conclusion} concludes. \Cref{app:topology} provides a brief convex-analytic motivation for mixtures of pricing measures.

\section{\sc{The Model}}\label{sec:Model}
Let $(\Omega,\mathfrak{F},\mathbb{F},P)$ be a complete filtered probability space satisfying the usual conditions, where $\mathbb{F}=\{\mathfrak{F}_t\}_{0\leq t\leq T}$. The paper uses $Y_t$ as the wealth or return state entering marginal utility. A more general economic state vector may be denoted by $X_t$, but whenever Arrow--Pratt curvature is computed the scalar state is $Y_t=g(X_t)$. Thus $Y_t$ is the payoff-relevant projection of the general state. Let $\sigma_t=\sigma(Y_t)$ denote the state-volatility process relevant for the diffusion representation below. In the empirical section, its option-implied counterpart is estimated from the implied-volatility surface. Let
\begin{equation}\label{eq:AP_process_def}
  A_t=A(Y_t)=-\frac{U^{\prime\prime}(Y_t)}{U^\prime(Y_t)}
\end{equation}
denote the Arrow--Pratt absolute risk-aversion process. Both $\sigma_t$ and $A_t$ are assumed to be progressively measurable.
Throughout the martingale-density construction, $A_t$ may be signed; positive values correspond to local risk aversion and negative values to local risk seeking, while admissibility is imposed on the stochastic-exponential kernel $A_t\sigma_t$.

\tab The object needed below is not pointwise equality between volatility and risk attitude. The paper uses a reduced-form measurable representation
\begin{equation}\label{eq:stoch_equiv}
  A_t=f(\sigma_t)+\varepsilon_t,
\end{equation}
where $f$ is a Borel measurable function and $\varepsilon_t$ is an innovation with $E[\varepsilon_t\mid\sigma_t]=0$ when the conditional expectation exists. This is an economic restriction rather than a topological one: implied volatility is treated as a market summary of state prices, while $A_t$ is the risk-attitude process consistent with those state prices. The representation is used below only to motivate the risk-management interpretation; the EBMM existence result does not require exact equality between $A_t$ and $f(\sigma_t)$.

\begin{assumption}[Risk-attitude dynamics]\label{ass:risk_attitude_dynamics}
The Arrow--Pratt process is recurrent and strong Markov. Recurrence means that local excursions of risk attitude return to the admissible risk region; the strong Markov property means that future evolution depends on the current state rather than on the distant past.\footnote{See \cite{McCau2007} on It\^o processes with finite memory.}
\end{assumption}
\subsection{Extended Samuelson--Merton ``util-prob'' options pricing}\label{subsec:extended_samuelson_merton}

\tab Following \cite{SamuelsonMerton1969}, let $Z=(Z_1(\omega),Z_2(\omega),\ldots,Z_n(\omega))$ denote terminal gross returns for $n$ assets and let $P$ be their physical probability distribution. Let $\mathbb{F}=\{\mathfrak{F}_t\}_{t\geq0}$ be a filtration of $\mathfrak{F}$. To connect this one-period return vector to the continuous-time notation used below, write $S_t$ for an asset price and $X_t=\log(S_t/S_0)$ for its log-return state. Under the lognormal benchmark,\footnote{See \citet[pp. ~273, 275]{Hull2006}}
 it is known  that if the stock price process is given by
\begin{equation}\label{Sproc}
\frac{dS(t,\omega)}{S(t,\omega)} =\mu dt+\sigma dW(t,\omega)
\end{equation}
then It\^o's formula gives the log-return process
\begin{equation}\label{Zproc}
dX(t,\omega)=\left(\mu-\frac{1}{2} \sigma^2 \right)dt+\sigma dW(t,\omega).
\end{equation}
Terminal gross returns can then be written as $Z_T=S_T/S_0=\exp(X_T)$. This distinction keeps the Samuelson--Merton integration variable $Z$ as a return level while using $X_t$ for the continuous-time log-return state.
The investor's problem constitutes maximizing the expected utility $U \in C^{2+\delta} (Z), \, \delta >0$, on her portfolio. For exposition, we assume that for t-fixed the investor's asset allocation scheme is $(w_1,w_2,\dotsc,w_n), \sum_{j=1}^n w_j =1, \, w_j \geq 0 $, and one of the assets is numeraire.\footnote{If the portfolio weights are not fixed, then we have a more complicated Hamilton-Jacobi-Bellman equation to solve and the salient point we are trying to make will be buried in technical detail. Arguably, we can view the investor's problem as a series of 1-period optimization, in which case the optimal allocation $w^*_{} (\omega)$ in each period, when indexed by $t\in[0,T] $, can be treated as a sample from the portfolio process $w(t,\omega),\mathfrak{F}_t; 0 \leq t \leq T< \infty $ where one asset is risk-free and used as numeraire.}
The investor's problem is
\begin{equation}\label{INVP}
  \begin{split}
    J(w;t,Z) &= E_P\left[U\left(\sum_{j=1}^n w_j Z_j(t)\right)\right] \\
    \text{s.t.}\quad &\sum_{j=1}^n w_j=1,\qquad w_j\geq0.
  \end{split}
\end{equation}
Let $W_t(w)=\sum_{j=1}^n w_jZ_j(t)$ and suppose that $U^\prime(W_t(w^*))$ is integrable and strictly positive. With the Kuhn--Tucker Lagrangian
\begin{equation}\label{eq:Lagrangian_clean}
  \mathcal{L}(w,\kappa,\mu)=E_P[U(W_t(w))]-\kappa\left(\sum_{j=1}^n w_j-1\right)+\sum_{j=1}^n\mu_jw_j,
\end{equation}
where $\mu_j\geq0$ and $\mu_jw_j^*=0$, the first-order condition is
\begin{equation}\label{ShadPri}
  E_P\left[Z_j(t)U^\prime(W_t(w^*))\right]=\kappa-\mu_j.
\end{equation}
For an interior allocation, $\mu_j=0$ and the marginal-utility weighted expected return is equalized across assets. With binding no-short-sale constraints, the shadow prices $\mu_j$ create wedges across assets.

\tab The Samuelson--Merton ``util-prob'' density is therefore the normalized marginal-utility density
\begin{equation}\label{RadNyk}
  \frac{dQ^U}{dP}=\frac{U^\prime(W_t(w^*))}{E_P[U^\prime(W_t(w^*))]}.
\end{equation}
It is a probability density if $0<E_P[U^\prime(W_t(w^*))]<\infty$, and it is equivalent to $P$ if $U^\prime(W_t(w^*))>0$ almost surely. Under this density,
\begin{equation}\label{eq:Q_expectation_clean}
  E_{Q^U}[Z_j(t)]=\frac{\kappa-\mu_j}{E_P[U^\prime(W_t(w^*))]}.
\end{equation}
Thus the economic content of the ``util-prob'' transformation is precise: marginal utility reweights states, while the Kuhn--Tucker multipliers record the effect of portfolio constraints.

\subsection{A recurrent benchmark for the Arrow--Pratt risk process}\label{subsec:LocalTime}
The following proposition gives a mathematically consistent recurrent benchmark for the Arrow--Pratt risk process.
\noindent The EBMM existence theorem below does not require this specific Ornstein--Uhlenbeck structure. The proposition is included only as a tractable recurrent benchmark for simulations and for the risk-management interpretation of locally mean-reverting risk attitudes.
\begin{prop}[Time-changed Arrow--Pratt risk process]\label{prop:APriskProcess}
Let $V$ be a Brownian motion and define
\begin{equation}\label{APtimeChange}
  R_t=e^{-t}V_{e^{2t}-1},\qquad t\geq0.
\end{equation}
Then $R_t$ is an Ornstein--Uhlenbeck diffusion satisfying
\begin{equation}\label{eq:OU_AP}
  dR_t=-R_t\,dt+\sqrt{2}\,dB_t,
\end{equation}
for a Brownian motion $B$. Its infinitesimal generator is
\begin{equation}\label{APgenerator}
  \mathcal{L}f(x)=-x f^\prime(x)+f^{\prime\prime}(x).
\end{equation}
Consequently $R_t$ is recurrent and Markov. When the application restricts attention to locally risk-averse agents, a positive Arrow--Pratt process can be obtained from $R_t$ by an increasing positive transform, for example $A_t=\exp(R_t)$. The later smile/smirk decomposition allows signed local curvature, with $A_t<0$ interpreted as a risk-seeking or speculative state component.
\end{prop}
\begin{proof}
Let $u(t)=e^{2t}-1$. Then $du=2e^{2t}dt$ and $dV_{u(t)}=\sqrt{u'(t)}\,dB_t=\sqrt{2}e^t dB_t$ for a Brownian motion $B$ under the deterministic time-change theorem. Since $R_t=e^{-t}V_{u(t)}$, It\^o's product rule gives $dR_t=-e^{-t}V_{u(t)}dt+e^{-t}\sqrt{2}e^t dB_t=-R_t\,dt+\sqrt{2}\,dB_t$. The generator in \eqref{APgenerator} follows from the drift coefficient $-x$ and diffusion coefficient $\sqrt{2}$. The one-dimensional Ornstein--Uhlenbeck diffusion is recurrent because its invariant Gaussian law has full support and the process is mean reverting.
\end{proof}

\section{\sc{Equivalent Behavioral Martingale Measure}}\label{sec:EBMM}
\subsection{Measure-theoretic motivation}\label{subsec:measure_theoretic_motivation}
The EBMM construction below only needs the standard measure-theoretic insight that admissible pricing measures are probability measures equivalent to the physical measure and represented by positive martingale density processes. The convex-analytic discussion in \Cref{app:topology} is motivational background for mixtures of pricing measures, not a foundation for the equivalence result. The main argument in the proof proceeds directly through Radon--Nikodym densities, stochastic exponentials, and Girsanov kernels.

\subsection{Martingale representation of equivalent behavioral density processes}\label{subsec:MartingaleRepresentationEBMM}
Consistent with \cref{sec:Model}, let $Y_t$ denote the wealth or return state entering marginal utility and suppose
\begin{equation}\label{eq:Yprocess}
  dY_t=b_t\,dt+\sigma_t\,dW_t,
\end{equation}
where $b_t$ and $\sigma_t$ are progressively measurable and $\int_0^T\sigma_t^2dt<\infty$ almost surely. Assume $U\in C^3$, $U^\prime>0$, and define
\begin{equation}\label{eq:APdefClean}
  A_t=-\frac{U^{\prime\prime}(Y_t)}{U^\prime(Y_t)},\qquad
  a_t=\frac{U^{\prime\prime\prime}(Y_t)}{U^\prime(Y_t)}.
\end{equation}
The term $U'''/U'$ captures the curvature channel associated with precautionary saving and prudence; see \citet{Sandmo1970} and the formal prudence measure in \citet{Kimball1990}. It\^o's formula gives
\begin{equation}\label{ItoMU}
  \frac{dU^\prime(Y_t)}{U^\prime(Y_t)}
  =\left[-A_t b_t+\frac{1}{2}a_t\sigma_t^2\right]dt-A_t\sigma_t\,dW_t.
\end{equation}
Equivalently,
\begin{equation}\label{eq:logMU}
  d\log U^\prime(Y_t)
  =\left[-A_t b_t+\frac{1}{2}(a_t-A_t^2)\sigma_t^2\right]dt-A_t\sigma_t\,dW_t.
\end{equation}
Define the compensator
\begin{equation}\label{eq:MUcompensator}
  \Gamma_t=-A_t b_t+\frac{1}{2}(a_t-A_t^2)\sigma_t^2.
\end{equation}
Then
\begin{equation}\label{EBMM}
  L_t=\frac{U^\prime(Y_t)}{U^\prime(Y_0)}
      \exp\left(-\int_0^t \Gamma_s\,ds-\frac{1}{2}\int_0^t A_s^2\sigma_s^2\,ds\right)
      =\mathcal{E}\left(-\int_0^t A_s\sigma_s\,dW_s\right).
\end{equation}
Here $\mathcal{E}(\cdot)$ denotes the Dol\'eans--Dade stochastic exponential. For a continuous Brownian local martingale $N_t=\int_0^t\theta_s\,dW_s$,
\begin{equation}\label{eq:DoleansDadeDefinition}
  \mathcal{E}(N)_t
  =
  \exp\left(N_t-\frac{1}{2}\langle N\rangle_t\right)
  =
  \exp\left(\int_0^t\theta_s\,dW_s-\frac{1}{2}\int_0^t\theta_s^2\,ds\right).
\end{equation}
The second equality follows because \eqref{eq:logMU} implies $U^\prime(Y_t)/U^\prime(Y_0)=\exp\{\int_0^t \Gamma_s\,ds-\int_0^t A_s\sigma_s dW_s\}$. Hence the additional quadratic-variation correction in \eqref{EBMM} is essential. Without it the process is the exponential of a local martingale, but not the Dol\'eans--Dade exponential.

\begin{thm}[Existence and equivalence of the EBMM]\label{thm:EBMM_equivalence}
Assume $U\in C^3$, $U^\prime(Y_t)>0$ almost surely, $\int_0^T A_s^2\sigma_s^2ds<\infty$ almost surely, and
\begin{equation}\label{eq:NovikovClean}
  E_P\left[\exp\left(\frac{1}{2}\int_0^T A_s^2\sigma_s^2\,ds\right)\right]<\infty.
\end{equation}
Then $L_t$ in \eqref{EBMM} is a strictly positive true $P$-martingale with $E_P[L_T]=1$. The measure $Q^B$ defined by $dQ^B/dP=L_T$ is equivalent to $P$ on $\mathfrak{F}_T$.
\end{thm}
\begin{proof}
Novikov's condition in \eqref{eq:NovikovClean} implies that the stochastic exponential $\mathcal{E}(-\int_0^t A_s\sigma_s dW_s)$ is a true martingale with unit expectation. Stochastic exponentials are strictly positive for finite quadratic variation, so $L_T>0$ almost surely. Hence $Q^B\ll P$. Conversely, if $P(A)>0$, then $Q^B(A)=E_P[L_T\mathbf{1}_A]>0$ because $L_T>0$ almost surely; therefore $P\ll Q^B$. The two measures have the same null sets, which proves equivalence.
\end{proof}
The Novikov condition is sufficient but not necessary. It is strong when either $A_t$ or $\sigma_t$ is unbounded. It is satisfied, for example, when $A_t\sigma_t$ is bounded on $[0,T]$. A standard weakening is Kazamaki's criterion: for the local martingale $N_t=-\int_0^t A_s\sigma_s\,dW_s$, uniform integrability of $\exp\{N_t/2\}$ as a submartingale is sufficient for $\mathcal{E}(N)_t$ to be a true martingale. The results below therefore require a true stochastic-exponential density; Novikov is used as the clean sufficient condition for exposition.

\tab Use the standard convention in which the Brownian risk-price process is
\begin{equation}\label{eq:lambda_def}
  \lambda_t=A_t\sigma_t.
\end{equation}
Then the behavioral density can be written as
\begin{equation}\label{EMM}
  L_t=\mathcal{E}\left(-\int_0^t \lambda_s\,dW_s\right)
  =\exp\left(-\int_0^t \lambda_s\,dW_s-\frac{1}{2}\int_0^t \lambda_s^2\,ds\right).
\end{equation}
This is the precise sense in which the EBMM embeds Arrow--Pratt risk attitude inside the martingale density: the density is generated by the negative stochastic integral of risk attitude times state volatility.

\subsection{Arrow--Pratt curvature, prudence, and martingale admissibility}\label{subsec:IdentifyingRestrictionsOnMU}
Pratt's regularity condition gives the link between the curvature of marginal utility and the dynamics of the Arrow--Pratt process. For a deterministic state $x$,
\begin{equation}\label{PrattRateCond}
  A^\prime(x)=A^2(x)-\frac{U^{\prime\prime\prime}(x)}{U^\prime(x)}.
\end{equation}
It is therefore useful to define the dimensionless Pratt adjustment index
\begin{equation}\label{eq:alphaClean}
  \alpha(x)=\frac{U^{\prime\prime\prime}(x)}{U^\prime(x)A^2(x)},\qquad A(x)\neq0.
\end{equation}
Then
\begin{equation}\label{eq:rprimeAlpha}
  A^\prime(x)=(1-\alpha(x))A^2(x).
\end{equation}
Thus $\alpha$ is not equal to prudence $U^{\prime\prime\prime}/U^\prime$; rather, it rescales that prudence based curvature by $A^2$. In the stochastic notation of \eqref{eq:APdefClean}, $a_t=U^{\prime\prime\prime}(Y_t)/U^\prime(Y_t)$ and $\alpha_t=a_t/A_t^2$ whenever $A_t\neq0$. If absolute risk aversion is locally increasing, then $\alpha(x)<1$; if $U^{\prime\prime\prime}\geq0$, then $\alpha(x)\geq0$.

\tab When the state is stochastic, a reduced-form Arrow--Pratt diffusion may be written as
\begin{equation}\label{APriskRate}
  dA_t=(1-\alpha_t)A_t^2\,dt+\eta_t\,dB_t,
\end{equation}
where $B$ is a Brownian motion and $\eta_t$ is progressively measurable. Equation \eqref{APriskRate} is an SDE; the formal expression $dA_t/dt$ is not used because Brownian paths are nowhere differentiable. This reduced-form diffusion is not needed for \Cref{thm:EBMM_equivalence}; it is a local parametrization of Arrow--Pratt dynamics used to interpret the prudence adjustment factor. Combining \eqref{eq:rprimeAlpha} with the compensated EBMM density in \eqref{EBMM} shows that the finite-variation term to be compensated is
\begin{equation}\label{AdmisRes}
  \int_0^t\left[-A_s b_s+\frac{1}{2}\sigma_s^2 A_s^2(\alpha_s-1)\right]ds.
\end{equation}
The condition is not that marginal utility itself must be constant. Rather, the normalized and compensated marginal-utility process must be a true martingale. Under Novikov's condition \eqref{eq:NovikovClean}, the admissible EBMM density is
\begin{equation}\label{UprimeCond}
  L_t=\frac{U^\prime(Y_t)}{U^\prime(Y_0)}
  \exp\left\{\int_0^t\left[A_s b_s-\frac{1}{2}\alpha_s\sigma_s^2 A_s^2\right]ds\right\}.
\end{equation}
Equation \eqref{UprimeCond} follows from substituting $a_t=\alpha_tA_t^2$ into $\Gamma_t$ and combining the It\^o drift correction with the Dol\'eans--Dade quadratic-variation correction.
This corrected expression is the main mathematical restriction: behavioral weighting by marginal utility becomes an equivalent martingale measure only after the drift generated by risk aversion and prudence is compensated.

\subsection{Weighted Average Martingale Measure for Heterogeneous Returns}\label{subsec:WeightedAvgMartingaleMeasure}
In this subsection, we apply the weighted-average martingale idea of \cite{KazTsu1967} to finite mixtures of equivalent behavioral martingale densities for heterogeneous returns. For each asset $j$, let
\begin{equation}
  M_j(t)=\mathcal{E}\left(-\int_0^t A_{j,s}\sigma_{j,s}\,dW_{j,s}\right)
\end{equation}
be the behavioral stochastic exponential, and assume the corresponding Novikov condition holds. Let $v_j>0$ be raw behavioral weights and define normalized weights
\begin{equation}\label{eq:normalizedWeights}
  \pi_j^{(N)}=\frac{v_j}{W_N},\qquad W_N=\sum_{k=1}^N v_k.
\end{equation}
Then $\sum_{j=1}^N\pi_j^{(N)}=1$. Economically, $\pi_j^{(N)}$ may be interpreted as the market share of agent type $j$, the probability that the marginal trader is of type $j$, or a liquidity/trading-volume weight attached to that component density. The point is not that the raw $v_j$ are primitive probabilities, but that the normalized mixture must be a probability density before it can define an equivalent measure. The finite behavioral mixture density is
\begin{equation}\label{EBMMwgt1}
  \frac{dQ^{(N)}}{dP}=\sum_{j=1}^N\pi_j^{(N)}M_j(T).
\end{equation}
Because each $M_j$ is nonnegative with $E_P[M_j(T)]=1$, the mixture in \eqref{EBMMwgt1} is also nonnegative and integrates to one. The convex-analytic note in  \Cref{app:topology} gives a short motivation for reading this construction as a finite mixture of admissible pricing densities, while the admissibility result below rests on normalization and the martingale property. The associated density process is
\begin{equation}\label{EBMMwgt2}
  L_t^{(N)}=E_P\left[\frac{dQ^{(N)}}{dP}\bigg|\mathfrak{F}_t\right]
          =\sum_{j=1}^N\pi_j^{(N)}M_j(t),
\end{equation}
provided the weights are deterministic or $\mathfrak{F}_0$-measurable. Thus the right-hand side is a weighted martingale.

\begin{thm}[Weighted EBMM]\label{thm:weightedEBMM}
Suppose $M_j(t)$, $j=1,\ldots,N$, are nonnegative $P$-martingales with $M_j(0)=1$ and $E_P[M_j(T)]=1$. If $\pi_j^{(N)}\geq0$ and $\sum_{j=1}^N\pi_j^{(N)}=1$, then $L_t^{(N)}=\sum_{j=1}^N\pi_j^{(N)}M_j(t)$ is a nonnegative $P$-martingale with $L_0^{(N)}=1$. Hence $dQ^{(N)}/dP=L_T^{(N)}$ defines an equivalent behavioral martingale measure whenever $L_T^{(N)}>0$ almost surely.
\end{thm}
\begin{proof}
Linearity of conditional expectation gives
\begin{equation*}
E_P[L_t^{(N)}\mid\mathfrak{F}_s]=\sum_{j=1}^N\pi_j^{(N)}E_P[M_j(t)\mid\mathfrak{F}_s]=\sum_{j=1}^N\pi_j^{(N)}M_j(s)=L_s^{(N)}
\end{equation*}
for $s\leq t$. Nonnegativity and unit expectation follow from the same properties of the component densities and from the normalization of the weights.
\end{proof}
This result is the finite-mixture analogue of the weighted-average martingale results of \cite{KazTsu1967}. It also clarifies the required normalization: raw behavioral weights may be arbitrary positive numbers, but the density must use normalized weights.
\section{\sc{Risk Management}}\label{sec:RiskManagement}
\subsection{Cox--Ross risk-premium identity and EBMM cancellation}\label{subsec:cox_ross_identity}
According to \citet{SamuelsonMerton1969}, marginal utility weights probabilities. The corrected EBMM construction shows that this weighting becomes a martingale density only after the finite-variation term in \eqref{AdmisRes} has been compensated. The Cox--Ross cancellation claim can be stated directly in terms of the Brownian risk price.

\tab Let the underlying satisfy
\begin{equation}\label{eq:S_under_Q}
  \frac{dS_t}{S_t}=\mu_S(t)dt+\sigma_S(t)dW_t,
\end{equation}
and let an option or derivative value $C_t=C(S_t,t)$ satisfy
\begin{equation}\label{eq:C_under_Q}
  \frac{dC_t}{C_t}=\mu_C(t)dt+\sigma_C(t)dW_t,
\end{equation}
where $\sigma_C(t)=\Delta_t S_t\sigma_S(t)/C_t$ when It\^o's formula is applied to a smooth derivative price. If both assets are priced by the same equivalent martingale measure with Brownian risk-price process $\lambda_t$, then their physical drift premia satisfy
\begin{equation}\label{eq:risk_premia_kernel}
  \mu_S(t)-r_f(t)=\sigma_S(t)\lambda_t,\qquad
  \mu_C(t)-r_f(t)=\sigma_C(t)\lambda_t.
\end{equation}
Dividing the two equations gives the \citet{CoxRoss1976} style cancellation identity
\begin{equation}\label{eq:CoxRossCancel}
  \frac{\mu_C(t)-r_f(t)}{\mu_S(t)-r_f(t)}=\frac{\sigma_C(t)}{\sigma_S(t)}
  =\frac{\Delta_t S_t}{C_t}.
\end{equation}
Thus the common martingale-density effect cancels from the ratio. The EBMM adds behavioral content by substituting $\lambda_t=A_t\sigma_S(t)$ for the underlying. Then
\begin{equation}\label{eq:EBMMRiskPremium}
  \mu_S(t)-r_f(t)=A_t\sigma_S^2(t),\qquad
  \mu_C(t)-r_f(t)=A_t\sigma_S(t)\sigma_C(t).
\end{equation}
Let $A_t^S$ denote the Arrow--Pratt process inferred from the underlying and let $A_t^C$ denote the corresponding process inferred from the option. If an option-specific risk-attitude process $A_t^C$ is inferred from its own risk-price relation $\lambda_t=A_t^C\sigma_C(t)$, equality of the pricing kernel implies
\begin{equation}\label{eq:APRatio}
  A_t^C\sigma_C(t)=A_t^S\sigma_S(t),\qquad
  \frac{A_t^C}{A_t^S}=\frac{\sigma_S(t)}{\sigma_C(t)}.
\end{equation}
Equations \eqref{eq:CoxRossCancel}--\eqref{eq:APRatio} are the precise version of the paper's cancellation claim: the density kernel is common, so it drops out of the option-underlying premium ratio, while the EBMM identifies that kernel with Arrow--Pratt risk attitude times volatility.

\begin{prop}[Cox--Ross cancellation under EBMM]\label{prop:CoxRossEBMM}
Suppose $S_t$ and $C_t$ are driven by the same Brownian motion and are priced by the equivalent measure generated by $L_t$ in \eqref{EBMM}. If $\sigma_S(t)$ and $\sigma_C(t)$ are nonzero, then the risk-premium ratio is given by \eqref{eq:CoxRossCancel}. Under the EBMM compatibility condition, the same ratio can be written in terms of Arrow--Pratt risk-attitude components as \eqref{eq:APRatio}.
\end{prop}
\begin{proof}
Under Girsanov's theorem for $L_t=\mathcal{E}(-\int_0^t\lambda_s\,dW_s)$, the Brownian motion under the pricing measure is
\[
  W_t^Q=W_t+\int_0^t\lambda_s\,ds,\qquad dW_t=dW_t^Q-\lambda_t\,dt.
\]
Substituting this relation into \eqref{eq:S_under_Q} and \eqref{eq:C_under_Q} and imposing discounted martingale dynamics under $Q$ gives \eqref{eq:risk_premia_kernel}. Division proves \eqref{eq:CoxRossCancel}. The EBMM identity $\lambda_t=A_t\sigma_t$ then gives \eqref{eq:EBMMRiskPremium}; applying the same common kernel to the option and underlying gives \eqref{eq:APRatio}.
\end{proof}
\noindent In multifactor or incomplete settings, the ratio generalizes to equality of exposures to the common priced Brownian components; the scalar cancellation identity above is the one-factor benchmark.

\tab Put--call parity is not altered by EBMM when calls, puts, the underlying, and the bond are priced by the same admissible behavioral martingale density. The behavioral density changes the interpretation of the pricing kernel, but the payoff identity $(S_T-K)^+-(K-S_T)^+=S_T-K$ is pathwise, so the kernel cancels in the parity portfolio. Thus, for European options without dividends, the put--call parity relation is
\begin{equation}\label{eq:PutCallParity}
  C_t-P_t=S_t-Ke^{-r_f(T-t)},
\end{equation}
with the usual dividend adjustment when the underlying pays dividends. Apparent parity violations would therefore be interpreted not as new EBMM pricing formulas, but as evidence of frictions, segmentation, stale quotes, funding constraints, short-sale limits, or failure of the common-kernel admissibility condition. Similarly, a delta-hedged call is locally insulated from first-order underlying price risk: the option's first-order exposure is offset by the underlying, leaving the admissible pricing kernel to determine the drift-adjusted value. This is a trading implementation of risk-neutral valuation, not an assumption of risk-neutral preferences.

\subsection{Smile and smirk representation by risk-attitude components}\label{subsec:SmileAndSmirk}
This paper uses the term $\tau$-welded component for a smooth transition between positive and negative state-dependent risk-attitude components across moneyness. This is a modelling choice rather than a primitive theorem. It is used because the even--odd decomposition gives the minimal distinction needed to represent symmetric smile curvature and asymmetric smirk slope. Let $m=\log(K/S_t)$ denote log-moneyness and let $\tau=T-t$ denote time to maturity. A convenient parsimonious $\tau$-welded basis is
\begin{equation}\label{eq:tau_weld}
  A_\theta(m,\tau)=\beta_0(\tau)+A_{\mathrm{averse}}(\tau)[\cosh(m)-1]+A_{\mathrm{seek}}(\tau)\sinh(m),
\end{equation}
where $\theta=(\beta_0,\beta_+,\beta_-,\ldots)$ collects the parameters that determine the model-implied behavioral density, with $A_{\mathrm{averse}}(\tau)>0$ and $A_{\mathrm{seek}}(\tau)<0$ under the sign convention used in the empirical section. The term $[\cosh(m)-1]$ is even, zero at the money, and controls symmetric curvature. The term $\sinh(m)$ is odd and controls skew. A negative $A_{\mathrm{seek}}$ is the signed risk-seeking or speculative tilt. The functions $\cosh(m)$ and $\sinh(m)$ are not asserted to be unique. Any smooth moneyness function can be decomposed into even and odd parts.\footnote{Let $f$ be smooth on a domain symmetric around zero. Define $f_e(m)=\{f(m)+f(-m)\}/2$ and $f_o(m)=\{f(m)-f(-m)\}/2$. Then $f=f_e+f_o$, $f_e(-m)=f_e(m)$, and $f_o(-m)=-f_o(m)$, so $f_e$ is even and $f_o$ is odd. Smoothness is inherited from $f$ and $f(-m)$. The decomposition is unique because any $f=g+h$ with $g$ even and $h$ odd gives $f(-m)=g(m)-h(m)$; adding and subtracting $f(m)$ and $f(-m)$ recovers $g=f_e$ and $h=f_o$.} The hyperbolic pair is a parsimonious smooth basis: $\cosh(m)-1$ supplies an even curvature term with exponential tail growth, while $\sinh(m)$ supplies the corresponding odd tilt. The maturity functions $\beta_0(\tau)$, $A_{\mathrm{averse}}(\tau)$, and $A_{\mathrm{seek}}(\tau)$ can be treated as smooth term-structure loadings and estimated from option panels, or fixed in an illustrative calibration. Combining \eqref{eq:EBMMRiskPremium} with \eqref{eq:tau_weld} yields the implied-risk-premium surface
\begin{equation}\label{eq:smile_surface}
  \mu_S(m,\tau)-r_f=A_\theta(m,\tau)\sigma_S^2(m,\tau).
\end{equation}
A smile corresponds to the case in which the even component dominates; a smirk corresponds to the case in which the odd component dominates. This is a structural restriction that translates a recovered state-price kernel into a decomposed behavioral risk-attitude surface. Existing empirical smile models may fit the surface directly; the EBMM asks whether the fitted curvature can be read as an admissible marginal-utility density.

\Cref{fig:EBMM_arrow_pratt_surface} plots an illustrative three-dimensional Arrow--Pratt surface generated from smooth maturity loadings in \eqref{eq:tau_weld}. The figure makes explicit the object that the calibration seeks to recover: $A_\theta(m,\tau)$ varies jointly over log-moneyness and time to maturity. The surface height and light grid shading separate regions with higher and lower behavioral risk sensitivity. The surface is schematic, but it clarifies why the model is not simply fitting a single smile curve; it imposes a joint restriction across strike and maturity. \Cref{fig:EBMM_arrow_pratt_surface_contour} gives the corresponding two-dimensional projection. The projection makes it easier to read level sets and to see how slices through the surface generate smile- and smirk-type patterns.

\begin{figure}[!htbp]
\centering
\includegraphics[width=0.72\textwidth]{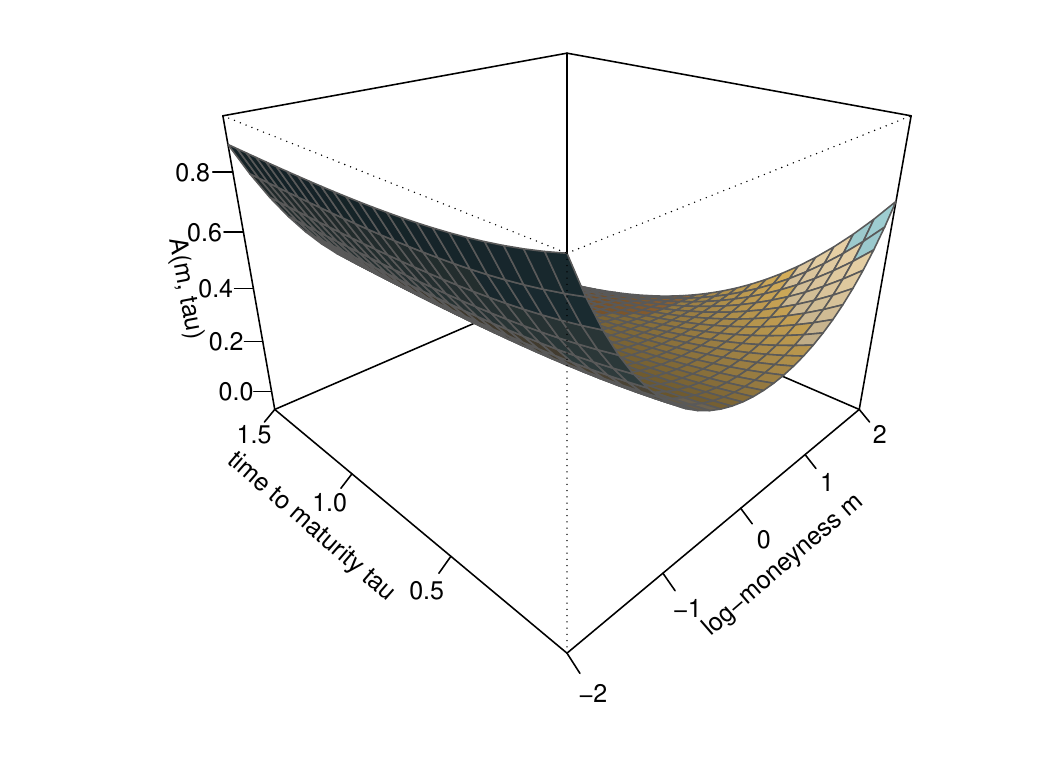}
\caption{Illustrative three-dimensional Arrow--Pratt risk-attitude surface $A_\theta(m,\tau)$. The surface is generated from smooth maturity loadings $\beta_0(\tau)$, $A_{\mathrm{averse}}(\tau)$, and $A_{\mathrm{seek}}(\tau)$ in the $\tau$-welded hyperbolic representation. Height and shading show behavioral risk sensitivity across log-moneyness and time to maturity.}
\label{fig:EBMM_arrow_pratt_surface}
\end{figure}

\begin{figure}[!htbp]
\centering
\includegraphics[width=0.72\textwidth]{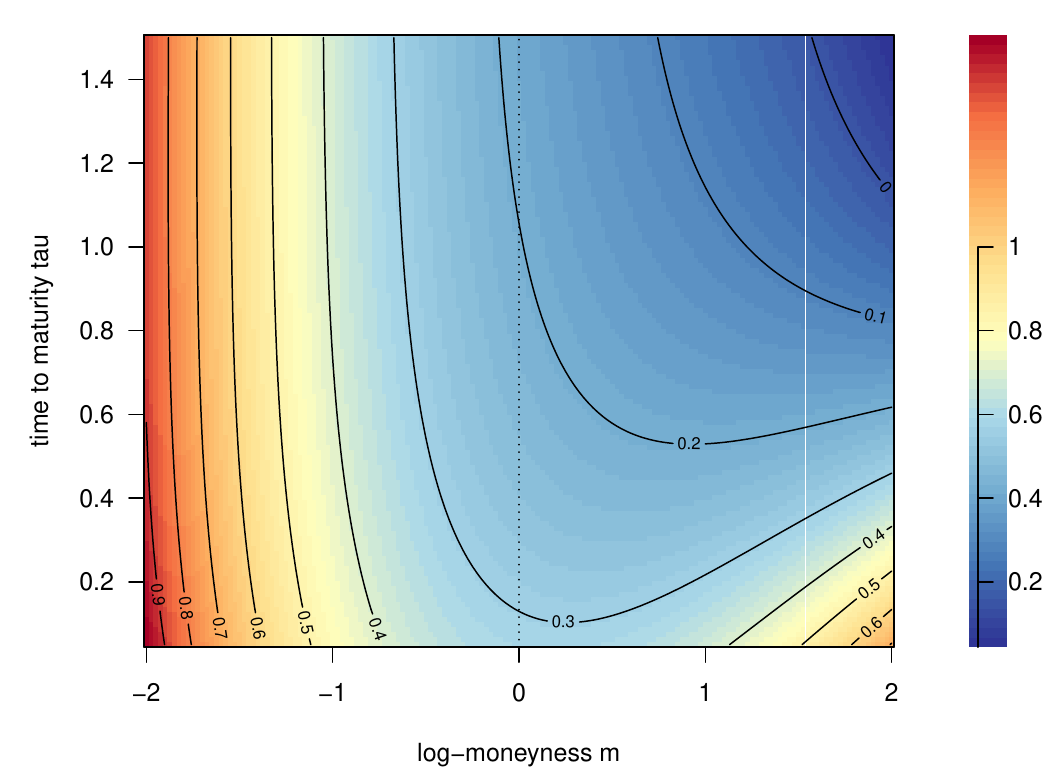}
\caption{Two-dimensional projection of the illustrative Arrow--Pratt risk-attitude surface $A_\theta(m,\tau)$. The contour lines are level sets of behavioral risk sensitivity across log-moneyness and time to maturity.}
\label{fig:EBMM_arrow_pratt_surface_contour}
\end{figure}

\begin{prop}[Smile/smirk decomposition]\label{prop:smile_smirk}
If $A_{\mathrm{averse}}(\tau)>0$ and $A_{\mathrm{seek}}(\tau)=0$, then the risk-attitude component in \eqref{eq:tau_weld} is symmetric in log-moneyness and can generate a volatility smile. If $A_{\mathrm{seek}}(\tau)\neq0$, then the odd hyperbolic component tilts the surface across moneyness and can generate a smirk.
\end{prop}
\begin{proof}
The function $\cosh(m)$ is even, so its contribution is the same at $m$ and $-m$. The function $\sinh(m)$ is odd, so its contribution changes sign across $m=0$. Therefore the first component creates symmetric curvature around at-the-money states, while the second creates asymmetric slope. Through \eqref{eq:smile_surface}, these components map into the curvature and skew of the implied-volatility surface.
\end{proof}

\subsection{Calibration to option-market regularities}\label{subsec:market_calibration}
The EBMM representation can be calibrated to option-market data by treating the compensated behavioral density as the object that links the physical distribution of returns to the option-implied distribution. For a fixed maturity $\tau$, let $p_\tau(x)$ denote the physical density of the terminal return or log-return state $x$, estimated from historical or filtered returns, and let $q_\tau(x)$ denote the risk-neutral density recovered from option prices. Under no arbitrage, the normalized state-price density ratio is
\begin{equation}\label{eq:empirical_kernel_ratio}
  \widehat L_\tau(x)=\frac{q_\tau(x)}{p_\tau(x)},\qquad
  \int \widehat L_\tau(x)p_\tau(x)dx=1.
\end{equation}
The discounted stochastic discount factor is $e^{-r_f\tau}\widehat L_\tau(x)$. Thus the empirical counterpart of the EBMM density in \eqref{EBMM} is observable once the physical and risk-neutral densities have been estimated. The risk-neutral density may be recovered from a smoothed option surface using the Breeden--Litzenberger second-derivative relation, or by a semiparametric density estimator as in the state-price recovery literature \citep{BreLitzen1978,AitSahaliaLo2000,FrassoEilers2022}.

\tab A direct calibration proceeds in five steps. First, collect a panel of European-style option prices across strikes $K_i$ and maturities $\tau_j$, and convert strikes into log-moneyness $m_i=\log(K_i/S_t)$. Second, fit an arbitrage-smoothed implied-volatility surface and recover $q_{\tau_j}(x)$ for each maturity; a submission-grade implementation should use an arbitrage-free surface construction such as stochastic volatility inspired (SVI) restrictions of \citet{GatheralJacquier2014}, stochastic alpha, beta, rho (SABR), or another benchmark with static-arbitrage controls. Third, estimate $p_{\tau_j}(x)$ from realized returns over the same horizon, using a nonparametric kernel estimator, a filtered GARCH or stochastic-volatility density, a realized-return bootstrap, or a Ross-recovery physical density. Fourth, form the empirical density ratio $\widehat L_{\tau_j}(x)$ in \eqref{eq:empirical_kernel_ratio}. Fifth, choose the EBMM parameters $\theta=(\beta_0,\beta_+,\beta_-,\ldots)$ so that the model-implied behavioral density matches the recovered density ratio.

\tab For the hyperbolic smile/smirk specification in \eqref{eq:tau_weld}, the cross-sectional calibration can be written as a minimum-distance problem:
\begin{equation}\label{eq:EBMM_calibration_objective}
  \widehat\theta
  =\arg\min_{\theta\in\Theta}
  \sum_{j}\sum_i \omega_{ij}
  \left[\log \widehat L_{\tau_j}(x_i)-\log L^{EBMM}_{\theta}(x_i,\tau_j)\right]^2,
\end{equation}
where $\omega_{ij}$ are liquidity or inverse-noise weights. The model-implied density is normalized by construction,
\begin{equation}\label{eq:EBMM_normalized_density_fit}
  L^{EBMM}_{\theta}(x,\tau)
  =\frac{\exp\{\ell_\theta(x,\tau)\}}
  {\int \exp\{\ell_\theta(u,\tau)\}p_\tau(u)du},
\end{equation}
with $\ell_\theta(x,\tau)$ chosen to be consistent with the compensated marginal-utility density in \eqref{EBMM}. In a one-period cross-sectional implementation, $\ell_\theta$ is the log behavioral kernel implied by the fitted Arrow--Pratt surface $A_\theta(m,\tau)$ and the observed volatility surface. In a dynamic implementation, it is obtained by discretizing the stochastic exponential $\mathcal{E}(-\int A_s\sigma_s dW_s)$ over the return horizon.

\tab The calibration delivers four diagnostics. The first is density fit: whether $L^{EBMM}_{\theta}$ matches the empirical pricing-kernel ratio $q_\tau/p_\tau$ across moneyness and maturity. The second is benchmark discipline: whether the EBMM representation remains competitive against fair arbitrage-controlled smile surfaces, not merely against unstable raw parametrizations. The third is admissibility: whether the fitted kernel satisfies positivity, normalization, and a sufficient martingale condition such as Novikov's criterion over the relevant horizon. For a discretized implementation, a transparent diagnostic analogue is
\begin{equation}\label{eq:empirical_novikov}
  \widehat{\mathcal{N}}_j
  =
  \sum_i \widehat w_{ij}
  \exp\left\{\frac{1}{2}\tau_j\widehat\lambda_{ij}^{\,2}\right\},
  \qquad
  \widehat\lambda_{ij}=\widehat A_\theta(m_{ij},\tau_j)\widehat\sigma_{\mathrm{imp},ij},
\end{equation}
where $\widehat w_{ij}$ are physical-density quadrature weights on the maturity slice. This is a finite-sample consistency screen, not a proof of the continuous-time Novikov condition. Failure is evidence against the fitted EBMM interpretation, while passage is only weak support for admissibility under the maintained discretization. The fourth diagnostic is economic decomposition: whether the recovered structural factors $\widehat A_{\mathrm{averse}}(\tau)$ or $\widehat A_{\mathrm{seek}}(\tau)$ dominate. A dominant even component corresponds to a smile-type regularity, while a dominant odd component corresponds to a smirk-type regularity. This makes the empirical test sharper than visual curve fitting: the EBMM is supported only if the option-implied density ratio can be represented by an admissible compensated marginal-utility density.

\tab The empirical null and alternative can be stated concretely. The null is that the recovered pricing-kernel ratio is either not representable by a positive normalized EBMM density or that any apparent smile/smirk fit violates the martingale-admissibility restrictions. The alternative is that a parameter vector $\theta$ exists such that $L^{EBMM}_{\theta}>0$, $\int L^{EBMM}_{\theta}dP=1$, the fitted kernel satisfies a sufficient stochastic-exponential condition such as Novikov's criterion on the sample horizon, and the EBMM fit is not dominated by standard smile benchmarks under out-of-sample pricing or density-error loss. The model would be falsified, in this implementation, by persistent negative or non-normalizable fitted densities, failure of the admissibility diagnostic, unstable signs of the even/odd components across dates and maturities, or inferior out-of-sample performance relative to arbitrage-free benchmark surfaces such as SVI.

\subsection{Option-panel data calibration}\label{subsec:thetadata_calibration}
As a first market-data check, we obtain data from ThetaData Terminal v3 for a panel of liquid index ETF and single-name options: SPY, QQQ, IWM, AAPL, MSFT, NVDA, AMZN, and TSLA. The corresponding underlying equity prices were obtained from FMPcloud. Before any surface fitting or density-ratio recovery or other data analysis, the option panel is subjected to a bid--ask liquidity filter more fully described below. This subsection should be read as a reproducible diagnostic stress test, not as a final empirical validation. Although the theoretical calibration protocol is stated for European-style options, the diagnostic panel uses listed U.S. ETF and single-name options. The exercise treats early-exercise effects as part of the empirical approximation; a production implementation should use European index options or explicitly correct for American exercise features. The current implementation uses multiple trade dates and listed expirations, pulls both end-of-day chains and historical implied volatilities, and keeps calls and puts separately in the fitted panel. It then fits the hyperbolic EBMM basis
\begin{equation}
  \sigma_{\mathrm{imp}}(m)
  =
  \beta_0+\beta_+[\cosh(m)-1]-\beta_-\sinh(m)+\varepsilon(m),
  \qquad m=\log(K/S),\label{eq:IV_cosh_sinh_rel}
\end{equation}
 to observed implied volatilities. 
 
 \tab As indicated above, the data was subjected to a bid--ask filter before data analysis. The bid--ask spread is a standard measure of option-market liquidity and quote noise
 \citep{Mayhew2002}; accordingly, and following the empirical option-pricing convention of using bid--ask midpoint quotes
 \citep{BakshiCaoChen1997}, the baseline sample applies a relative-spread filter before estimation. Here, we apply a light bid--ask screen before estimation. Quotes must satisfy $bid>0$, $ask>bid$, $7\leq\tau_{days}\leq365$, and $(ask-bid)/mid\leq0.50$, where $mid=(ask+bid)/2$. This screen removes crossed, zero-bid, and extremely stale or noisy quotes, while preserving most of the tail information in the liquid ETF and mega-cap single-name panel. Estimation uses midpoint quotes, and every first-stage benchmark receives the same spread-based liquidity weight
\begin{equation}
  \omega_i=\frac{1}{1+\left((ask_i-bid_i)/mid_i\right)^2}.\label{eq:Spread_liquidity_weight}
\end{equation}
Thus wider but still admissible tail quotes are downweighted rather than automatically discarded. The weighting adjustment is applied symmetrically to the EBMM hyperbolic regression, the quadratic regression, the raw SVI objective, and the arbitrage-penalized EBMM and SVI objectives; it is not an EBMM-specific tuning parameter. 

\tab The benchmark stage estimates the first-stage implied-volatility shape, call-only and put-only decompositions, and holdout errors on the filtered, spread-weighted panel. It then constructs a diagnostic proxy for $\log(q_\tau/p_\tau)$, estimates the physical density from a blend of overlapping-horizon kernel density estimates and an EWMA filtered-innovation bootstrap, and reports finite-sample Novikov and tail-instability screens.

\tab The coefficients in the first-stage implied-volatility regression should not be interpreted as structural Arrow--Pratt factors by themselves. To connect the exercise to \eqref{eq:tau_weld} and \eqref{eq:EBMM_calibration_objective}, the script adds a second-stage constrained minimum-distance fit of the normalized EBMM density kernel to the recovered $\log(q_\tau/p_\tau)$ diagnostic. The second stage estimates $\widehat A_{\mathrm{averse}}>0$ and $\widehat A_{\mathrm{seek}}<0$: $\widehat A_{\mathrm{averse}}$ measures symmetric Arrow--Pratt risk-aversion intensity, while $\widehat A_{\mathrm{seek}}$ measures the signed risk-seeking or speculative component of the odd moneyness tilt, consistent with the sign convention in \citet{HoltLaury2002}. The risk-neutral boundary is $\widehat A_{\mathrm{averse}}=0$, which is excluded by the theoretical restriction but can be approached numerically; such near-boundary estimates are interpreted empirically as local non-identification of the symmetric risk-aversion component rather than as a structural preference parameter equal to zero.
The two empirical stages used in the script are therefore
\begin{align}
  \sigma_{\mathrm{imp},ij}
  &=
  \beta_0+\beta_+[\cosh(m_{ij})-1]-\beta_-\sinh(m_{ij})+\varepsilon_{ij},
  \label{eq:EBMM_first_stage_iv}\\
  (\widehat A_{\mathrm{averse}},\widehat A_{\mathrm{seek}})
  &=
  \arg\min_{A_{\mathrm{averse}}>\,0,\;A_{\mathrm{seek}}<\,0}
  \sum_j\sum_i w_{ij}
  \left[
  \widehat\ell_{ij}
  -
  \left\{
    A_{\mathrm{averse}}[\cosh(m_{ij})-1]+A_{\mathrm{seek}}\sinh(m_{ij})
    -\log Z_j(A_{\mathrm{averse}},A_{\mathrm{seek}})
  \right\}
  \right]^2,
  \label{eq:EBMM_second_stage_AP}
\end{align}
where $\widehat\ell_{ij}=\log\{\widehat q_{\tau_j}(m_{ij})/\widehat p_{\tau_j}(m_{ij})\}$, $w_{ij}$ are physical-density weights, and
\begin{equation}
  Z_j(A_{\mathrm{averse}},A_{\mathrm{seek}})
  =
  \int
  \exp\{A_{\mathrm{averse}}[\cosh(u)-1]+A_{\mathrm{seek}}\sinh(u)\}
  \widehat p_{\tau_j}(u)\,du
\end{equation}
is the normalizing constant that enforces $\int L^{EBMM}_{A_{\mathrm{averse}},A_{\mathrm{seek}}}(u,\tau_j)\widehat p_{\tau_j}(u)du=1$ for each synchronized trade-date/maturity slice.
As a robustness check, we also estimate the same structural parameters by fitting the raw diagnostic ratio $\widehat q_\tau/\widehat p_\tau$ directly, without taking logs. Because the level ratio can be numerically explosive when $\widehat p_\tau$ is small, this check uses a truncated weighted-integral objective: observations are weighted by the physical-density quadrature weights, the empirical ratio is capped at the 99.5\% slice quantile, and residuals are evaluated with a Huber loss.\footnote{Huber loss is a robust error measure that treats small residuals like squared errors and large residuals like absolute errors, preventing outliers from dominating the fit; see \citet{Huber1964}.} The raw-ratio specification is intentionally severe and should be read as a stress test, not as the primary structural estimator. Stability of the signed risk-seeking component is informative, but near-boundary estimates of the positive risk-aversion component indicate that raw-level density ratios remain sensitive to tail singularities precisely in the deep-tail regions where option prices are most informative.

\begin{table}[!htbp]
\centering
\caption{Option-panel calibration sample.}
\label{tab:EBMM_theta_data_summary}
\begin{tabular}{lc}
\hline
Statistic & Value \\
\hline
Symbols & AAPL, AMZN, IWM, MSFT, NVDA, QQQ, SPY, TSLA \\
Trade dates & 20250102, 20250203, 20250303 \\
Expirations & 20250117, 20250221, 20250321 \\
Rows & 43642 \\
Moneyness range & [-0.350, 0.350] \\
Median implied volatility & 0.3392 \\
Implied-volatility range & [0.1066, 1.0355] \\
Median relative bid--ask spread & 0.0139 \\
Median underlying/forward & 384.2500 \\
\hline
\end{tabular}
\end{table}

\begin{table}[!htbp]
\centering
\caption{Bid--ask liquidity screen and spread-cutoff robustness.}
\label{tab:EBMM_liquidity_filter}
\begin{tabular}{lrrrrr}
\hline
Relative-spread cutoff & Rows & Median spread & EBMM & Arb.-EBMM & Arb.-SVI \\
\hline
25\% & 42359 & 0.0134 & 0.0898 & 0.0255 & 0.0298 \\
50\% & 43642 & 0.0139 & 0.0897 & 0.0265 & 0.0367 \\
100\% & 44854 & 0.0143 & 0.0916 & 0.0272 & 0.0465 \\
\hline
\end{tabular}
\begin{minipage}{0.88\textwidth}
\footnotesize Notes: The baseline screen keeps quotes with $bid>0$, $ask>bid$, $7\leq\tau_{days}\leq365$, and $(ask-bid)/mid\leq0.50$, where $mid=(ask+bid)/2$. The estimating equations use midpoints and liquidity weights $\omega_i=1/(1+((ask_i-bid_i)/mid_i)^2)$. The same spread weights are applied to the EBMM, quadratic, raw SVI, arbitrage-penalized EBMM, and arbitrage-penalized SVI objectives. The 25\% and 100\% rows are robustness checks around the 50\% baseline. RMSEs are holdout implied-volatility errors from the same fixed stratified split rule.
\end{minipage}
\end{table}

\begin{table}[!htbp]
\centering
\caption{First-stage implied-volatility shape fit.}
\label{tab:EBMM_theta_estimates}
\begin{tabular}{lc}
\hline
Parameter & Estimate \\
\hline
$\widehat\beta_0$ & 0.2856 \\
$\widehat\beta_+$ & 7.3340 \\
$\widehat\beta_-$ & -0.0666 \\
$R^2$ & 0.4184 \\
$N$ & 43642 \\
Median underlying/forward & 384.2500 \\
\hline
\end{tabular}
\begin{minipage}{0.88\textwidth}
\footnotesize Notes: The fitted panel uses the configured option panel across symbols, dates, and expirations. The dependent variable is implied volatility when available; otherwise the script uses an EOD midpoint proxy. The first-stage shape basis is $\beta_0+\beta_+[\cosh(m)-1]-\beta_-\sinh(m)$ with $m=\log(K/S)$. These beta coefficients describe the first-stage implied-volatility shape only and are not structural Arrow--Pratt parameters. Although $\widehat\beta_+$ is numerically large, $[\cosh(m)-1]$ is small over the retained moneyness range $[-0.350,0.350]$, so the fitted contribution remains inside the observed implied-volatility range. Structural Arrow--Pratt factors are estimated separately from the recovered density ratio.
\end{minipage}
\end{table}

\begin{table}[!htbp]
\centering
\caption{EBMM, quadratic, and SVI smile benchmarks.}
\label{tab:EBMM_model_comparison}
\begin{tabular}{lcccc}
\hline
Model & Train $R^2$ & Train RMSE & Test RMSE & Parameters \\
\hline
EBMM hyperbolic & 0.7747 & 0.0912 & 0.0897 & 3 \\
Arbitrage-penalized EBMM & 0.9815 & 0.0260 & 0.0265 & 3 \\
Quadratic smile & 0.7749 & 0.0912 & 0.0897 & 3 \\
Unconstrained raw SVI & -0.2258 & 0.2122 & 0.2101 & 5 \\
Arbitrage-penalized SVI & 0.9631 & 0.0368 & 0.0367 & 5 \\
\hline
\end{tabular}
\begin{minipage}{0.88\textwidth}
\footnotesize Notes: The test set is a stratified fixed random holdout from the multi-symbol, multi-date, multi-maturity option panel. EBMM and quadratic pooled regressions include symbol fixed effects when more than one ticker is configured. Arbitrage-penalized EBMM and arbitrage-penalized SVI are refit slice-by-slice using the same total-variance positivity and butterfly-density penalties. RMSEs are computed out of sample. The unconstrained SVI row is a raw parametric benchmark; neither penalized row is a full production static-arbitrage surface.
\end{minipage}
\end{table}

\begin{table}[!htbp]
\centering
\caption{Summary of recovered Arrow--Pratt risk-attitude factors.}
\label{tab:EBMM_arrow_pratt_risk_attitudes_summary}
\begin{tabular}{lccc}
\hline
Statistic & $\widehat A_{\mathrm{averse}}>0$ & $\widehat A_{\mathrm{seek}}<0$ & Density RMSE \\
\hline
Pooled fit & 0.0002 & -4.0558 & 1.8046 \\
Mean across slices & 7.3912 & -4.3729 & 1.8046 \\
Median across slices & 0.0000 & -2.4308 & 1.8046 \\
25th percentile & 0.0000 & -6.9648 & 1.8046 \\
75th percentile & 0.0000 & -0.8186 & 1.8046 \\
Share printed as 0.0000 & 75.0\% & -- & -- \\
Share with $\widehat A_{\mathrm{seek}}<0$ & -- & 100.0\% & -- \\
\hline
\end{tabular}
\begin{minipage}{0.88\textwidth}
\footnotesize Notes: The table summarizes the second-stage constrained minimum-distance estimates reported slice-by-slice in the appendix. The fitted log kernel is proportional to $A_{\mathrm{averse}}[\cosh(m)-1]+A_{\mathrm{seek}}\sinh(m)$ and is normalized against the physical density, so the intercept is absorbed by density normalization. The theoretical restrictions are $A_{\mathrm{averse}}>0$ and $A_{\mathrm{seek}}<0$; the optimizer imposes the interior bounds $A_{\mathrm{averse}}\geq 10^{-8}$ and $A_{\mathrm{seek}}\leq -10^{-8}$. Estimates printed as 0.0000 are therefore boundary-adjacent positive estimates, not negative risk aversion. The dispersion of $\widehat A_{\mathrm{averse}}$ is highly skewed, so the median and percentile rows better describe the boundary behavior than the arithmetic mean.
\end{minipage}
\end{table}

\begin{table}[!htbp]
\centering
\caption{Finite-sample Novikov consistency-screen summary.}
\label{tab:EBMM_novikov_sample_diagnostic_summary}
\begin{tabular}{lccc}
\hline
Statistic & Mean Novikov term & Max Novikov term & Max $|\widehat\lambda|$ \\
\hline
Minimum & 1.0000 & 1.0000 & 0.0000 \\
Median & 1.0054 & 1.0581 & 1.0646 \\
Mean & 1.0151 & 1.1841 & 1.4744 \\
Maximum & 1.0859 & 2.5082 & 6.1089 \\
Passing slices & 48/48 & 48/48 & 48/48 \\
\hline
\end{tabular}
\begin{minipage}{0.90\textwidth}
\footnotesize Notes: The screen summarizes the slice-level values of $N_j^{-1}\sum_i\exp\{0.5\tau_j\widehat\lambda_{ij}^2\}$ with $\widehat\lambda_{ij}=\widehat A_\theta(m_{ij},\tau_j)\widehat\sigma_{ij}$. All slice-level values are finite, consistent with the maintained admissibility assumption. This is a finite-sample consistency screen, not a proof of the continuous-time Novikov condition. Its empirical value is asymmetric: explosive values would reject the EBMM interpretation for the fitted slice, while finite values show only that the discrete data are consistent with admissibility under the maintained approximation. The full slice-level table is reported in the appendix.
\end{minipage}
\end{table}

\begin{table}[!htbp]
\centering
\caption{Predictive regression for next-horizon realized variance.}
\label{tab:EBMM_predictive_realized_variance}
\begin{tabular}{lcccc}
\hline
Regressor & Estimate & Std. error & $t$-stat. & $p$-value \\
\hline
$\alpha$ & 0.24539 & 0.04367 & 5.619 & $<0.0001$ \\
$\widehat A^{\mathrm{averse}}_{t}$ & -0.00194 & 0.00138 & -1.409 & 0.1656 \\
$\widehat A^{\mathrm{seek}}_{t}$ & 0.01840 & 0.00665 & 2.767 & 0.0082 \\
\hline
\end{tabular}
\begin{minipage}{0.90\textwidth}
\footnotesize Notes: The dependent variable is annualized realized variance over the option-maturity-matched forward window, $RV_{t,t+\tau}=252n^{-1}\sum_{s=t+1}^{t+n} r_s^2$. The regression is $RV_{t,t+\tau}=\alpha+\beta_1\widehat A^{\mathrm{averse}}_{t}+\beta_2\widehat A^{\mathrm{seek}}_{t}+\epsilon_t$ on symbol/date/expiration slices with available post-date returns. $N=48$ and $R^2=0.153$. This is a deliberately simple predictive-content check, not a causal design.
\end{minipage}
\end{table}

\tab The benchmark evidence should be read conservatively and literally. In the pooled holdout exercise, raw EBMM and the quadratic smile have nearly identical test RMSE, while the unconstrained raw SVI fit is unstable and much worse.  \Cref{tab:EBMM_liquidity_filter} shows that the main penalized EBMM--SVI ranking survives stricter and looser relative-spread cutoffs of 25\% and 100\%, although the exact RMSEs move with the retained tail quotes. We report an arbitrage-penalized EBMM fit: it keeps the same hyperbolic EBMM shape but refits each synchronized slice with the same total-variance positivity and butterfly-density penalties used for the arbitrage-penalized SVI benchmark.\footnote{$\frac{\partial^2 C(K,T)}{\partial K^2} \ge 0$ is the no-butterfly-arbitrage condition. Both arbitrage-penalized specifications penalize parameter choices that imply negative total variance or negative risk-neutral density, equivalently violations of convexity of option prices across strikes. For SVI, the objective is similar to
	$\min_{\theta_j} \sum_{i \in T_j} \ell(\sigma_i^{\text{obs}}-\sigma_{\text{SVI}}(k_i;\theta_j))+\lambda_1P_{\text{pos}}(\theta_j)+\lambda_2P_{\text{bf}}(\theta_j)$.
	The arbitrage-penalized EBMM fit replaces $\sigma_{\text{SVI}}(k_i;\theta_j)$ with the hyperbolic EBMM curve and applies the same penalty terms to the implied total-variance grid.} This is the appropriate comparison because the no-arbitrage restriction belongs to the option-price surface generated by the fitted curve, not to the SVI family alone.  \Cref{tab:EBMM_model_comparison} shows that this correction matters. Raw EBMM remains close to the quadratic benchmark. In this diagnostic implementation, arbitrage-penalized EBMM has lower pooled holdout RMSE than both raw EBMM and arbitrage-penalized SVI, although the slice-level table remains more mixed and continues to favor arbitrage-penalized SVI in many individual symbol/right slices. The empirical claim is therefore sharper: EBMM is not a substitute for production volatility-surface construction, but when the same static-arbitrage penalty is imposed on the EBMM curve, the pooled holdout benchmark no longer supports the objection that EBMM only looks good because SVI was treated unfairly. Its continuing role is to provide an interpretable signed Arrow--Pratt decomposition and martingale-admissibility screen for a candidate density ratio. A full arbitrage-free SVI surface in the spirit of \citet{GatheralJacquier2014}, SABR, or a comparable static-arbitrage-controlled estimator remains the appropriate production benchmark. Using the  \Cref{tab:EBMM_model_comparison} holdout errors, we also compute paired two-one-sided equivalence tests on holdout squared errors, together with one-sided non-inferiority tests. With an equivalence band equal to 5\% of the relevant benchmark holdout MSE, raw EBMM and the quadratic benchmark are statistically equivalent. Against unconstrained raw SVI, raw EBMM has materially lower error. Against arbitrage-penalized SVI, raw EBMM fails non-inferiority, but arbitrage-penalized EBMM has materially lower pooled holdout error.

\tab The paired tests are computed on holdout squared-error differences $d_i=e_{i,candidate}^2-e_{i,bench}^2$, where the candidate model is either raw EBMM or arbitrage-penalized EBMM. The equivalence band is $[-\Delta,\Delta]$, with $\Delta=0.05\,MSE_{bench}$. Thus the TOST null is $H_0:\mu_d\leq-\Delta$ or $\mu_d\geq\Delta$, against $H_1:-\Delta<\mu_d<\Delta$. Raw EBMM and the quadratic benchmark are statistically equivalent. Raw EBMM fails non-inferiority against arbitrage-penalized SVI. However, when EBMM receives the same arbitrage penalty, the mean paired MSE difference relative to arbitrage-penalized SVI is $-0.000645$, with a 90 percent confidence interval $[-0.000705,-0.000585]$ and non-inferiority $p<0.0001$. Thus arbitrage-penalized EBMM has materially lower pooled holdout squared-error loss. The pooled and slice-level conclusions differ because the pooled loss weights observations rather than slices equally. The appendix slice diagnostics show that SVI wins more individual symbol/right/maturity slices, whereas arbitrage-penalized EBMM reduces squared-error loss more in the pooled observation-level comparison. The result should therefore be interpreted as pooled-loss evidence, not as uniform slice-level dominance. Thus the present empirical claim is not that EBMM uniformly beats standard smile parametrizations. Rather, the panel exercise shows that the EBMM decomposition can be estimated after bid--ask filtering, subjected to the same static-arbitrage penalty used for SVI, benchmarked honestly, and screened for admissibility on real option data.

\tab Two identification cautions are important. First, many estimates of $\widehat A_{\mathrm{averse}}$ lie at the numerical lower bound. These estimates should be read as boundary estimates: the data in these slices do not separately identify a positive symmetric curvature component once the signed odd component is included. The empirical evidence in this sample is therefore primarily about the signed smirk or risk-seeking component, not about robust estimation of symmetric Arrow--Pratt risk aversion. Accordingly, the current data should not be read as strong empirical support for a separately identified symmetric risk-aversion factor; that component is largely unidentified in this small panel. This should be read as an empirical finding about the panel, not as evidence that the theoretical positive-curvature restriction is irrelevant. Second, calls and puts are fitted separately to allow for right-specific liquidity, quote, exercise, or microstructure effects. The common-kernel restriction behind put--call parity is a joint hypothesis about equality of the implied $A_t$ across rights; it is diagnosed here but not imposed in the separate call-only and put-only fits.

\tab The risk-seeking result should also be interpreted locally. In options markets, risk seeking does not mean that investors globally prefer risk in every state. In the EBMM/Arrow--Pratt representation, it means that the recovered pricing kernel contains a region of locally negative curvature, $A(m)<0$, so marginal utility rises over that part of the state or moneyness space. Economically, this corresponds to states in which the market assigns extra marginal value to positively convex upside payoffs. Such a pattern is consistent with speculative demand, lottery preference, leverage constraints, or clientele effects, and may appear as rich short-dated out-of-the-money calls, elevated right-wing implied volatility, expensive call risk reversals, or poor subsequent physical expected returns on lottery-like calls. In density-ratio terms, the signature is that $\widehat q_\tau/\widehat p_\tau$ rises over favorable states, so upside states are priced more richly than their physical probabilities alone would justify. This interpretation is compatible with simultaneous risk aversion in downside states: the same option surface can reflect demand for crash insurance in the left tail and speculative demand for jackpot payoffs in the right tail.

\begin{figure}[!htbp]
\centering
\includegraphics[width=0.72\textwidth]{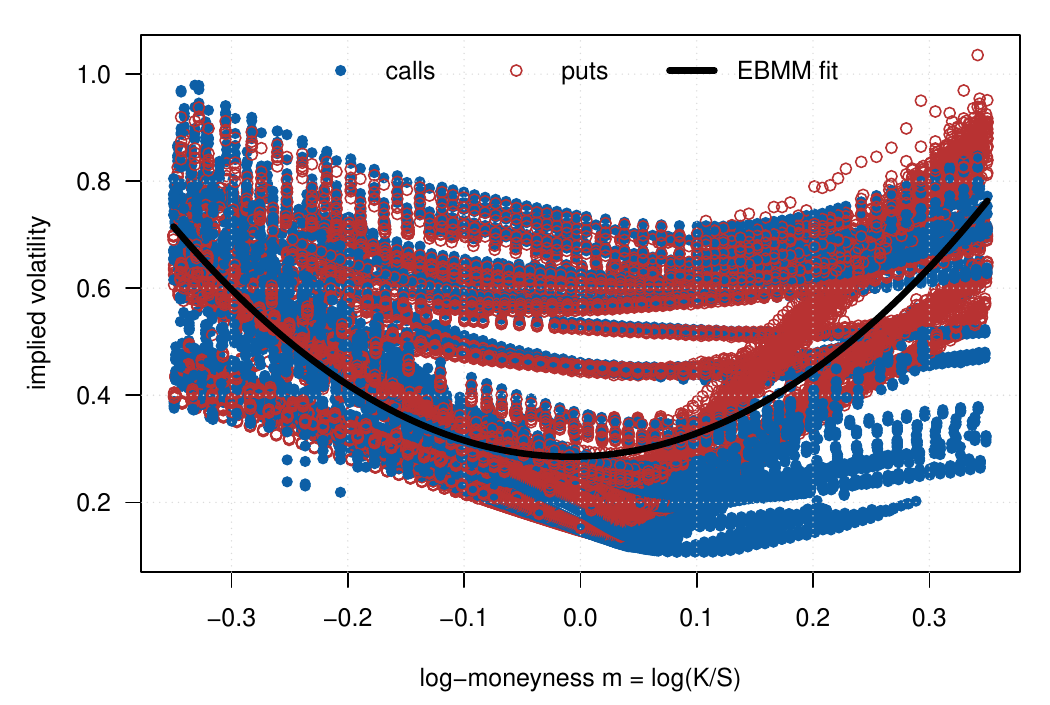}
\caption{ThetaData option-panel implied-volatility cross-section and fitted EBMM hyperbolic smile/smirk basis. Calls and puts are plotted against log-moneyness, with the fitted EBMM curve in black.}
\label{fig:EBMM_market_smile_fit}
\end{figure}

\begin{figure}[!htbp]
\centering
\includegraphics[width=0.72\textwidth]{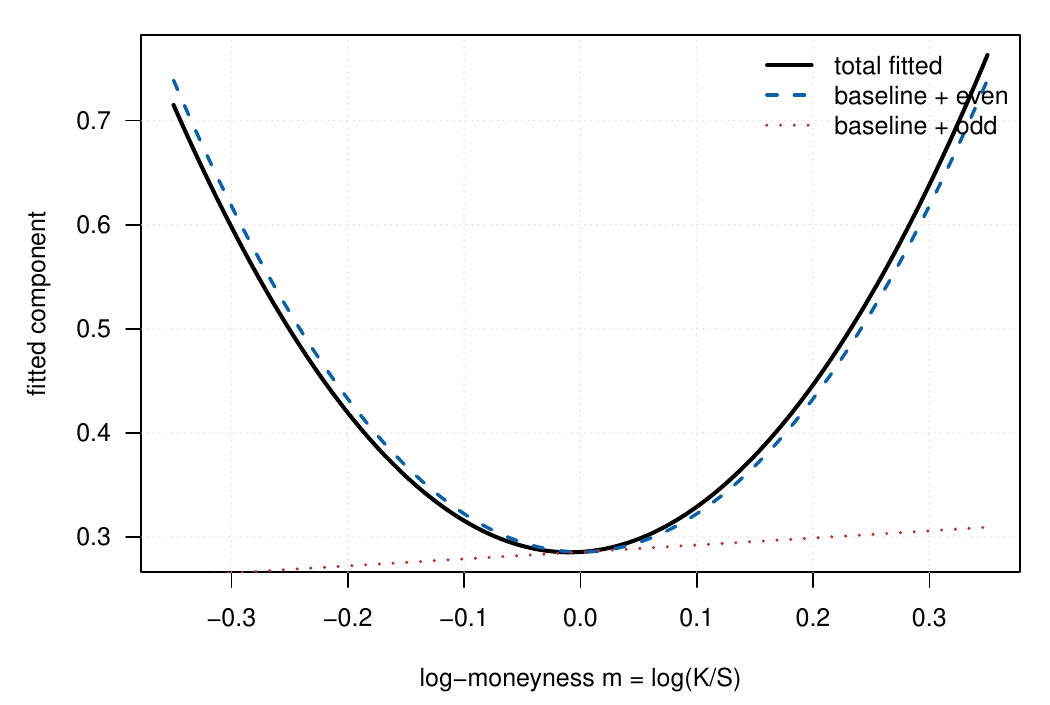}
\caption{Decomposition of the first-stage ThetaData panel EBMM implied-volatility curve into baseline, even smile, and odd smirk components. The structural $\widehat A_{\mathrm{averse}}$ and $\widehat A_{\mathrm{seek}}$ factors are estimated separately from the recovered density-ratio kernel.}
\label{fig:EBMM_market_decomposition}
\end{figure}

\tab The finite-sample Novikov screen reports whether the fitted $\widehat\lambda=\widehat A\widehat\sigma$ produces explosive exponential-square terms on each symbol/maturity slice, and the tail-instability screen flags slices where the estimated physical density is too small for raw density ratios to be trusted. All reported slice-level Novikov terms are finite when the screen column is marked Yes, consistent with the maintained admissibility assumption; this is a finite-sample consistency screen, not a proof of the continuous-time Novikov condition. The empirical value of the screen is asymmetric: explosive values would reject the EBMM interpretation for the fitted slice, while finite values only show that the discrete data are consistent with admissibility under the maintained approximation. The predictive regression in  \Cref{tab:EBMM_predictive_realized_variance} asks whether the recovered signed Arrow--Pratt factors forecast next-horizon realized variance. In this small panel, the signed risk-seeking component is positive and statistically significant, while the risk-aversion component is not statistically distinguishable from zero. Because $\widehat A_{\mathrm{seek}}<0$ by construction, the positive coefficient is a signed effect: a less negative risk-seeking component predicts higher realized variance. Equivalently, the risk-seeking intensity $-\widehat A_{\mathrm{seek}}>0$ has the opposite sign in this sample. The test has only 48 synchronized symbol/date/expiration slices and does not include market-wide controls such as VIX, term-spread variables, realized skew, or liquidity measures, so it should be read as a small-sample predictive check rather than a full forecasting study. 

\subsection{Illustrative smile and smirk calibration}\label{subsec:illustrative_smile_smirk_calibration}
A simple numerical illustration shows what the representation is designed to do. Normalize $\sigma_S(m,\tau)=0.20$ and $r_f=0$. If $A_{\mathrm{averse}}(\tau)=0.20$ and $A_{\mathrm{seek}}(\tau)=0$, then \eqref{eq:tau_weld} is even in $m$ and produces larger risk premia for deep in-the-money and out-of-the-money states than for at-the-money states. If $A_{\mathrm{averse}}(\tau)=0.10$ and $A_{\mathrm{seek}}(\tau)=-0.18$, the odd component tilts the surface so that one side of the moneyness distribution carries a larger behavioral premium. The preceding subsection gives the market-calibration counterpart. The present numerical example is intentionally schematic: it demonstrates that the EBMM decomposition can map symmetric curvature and asymmetric skew into risk-premium surfaces, but it does not use observed option prices.

\begin{figure}[!htbp]
\centering
\includegraphics[width=0.70\textwidth]{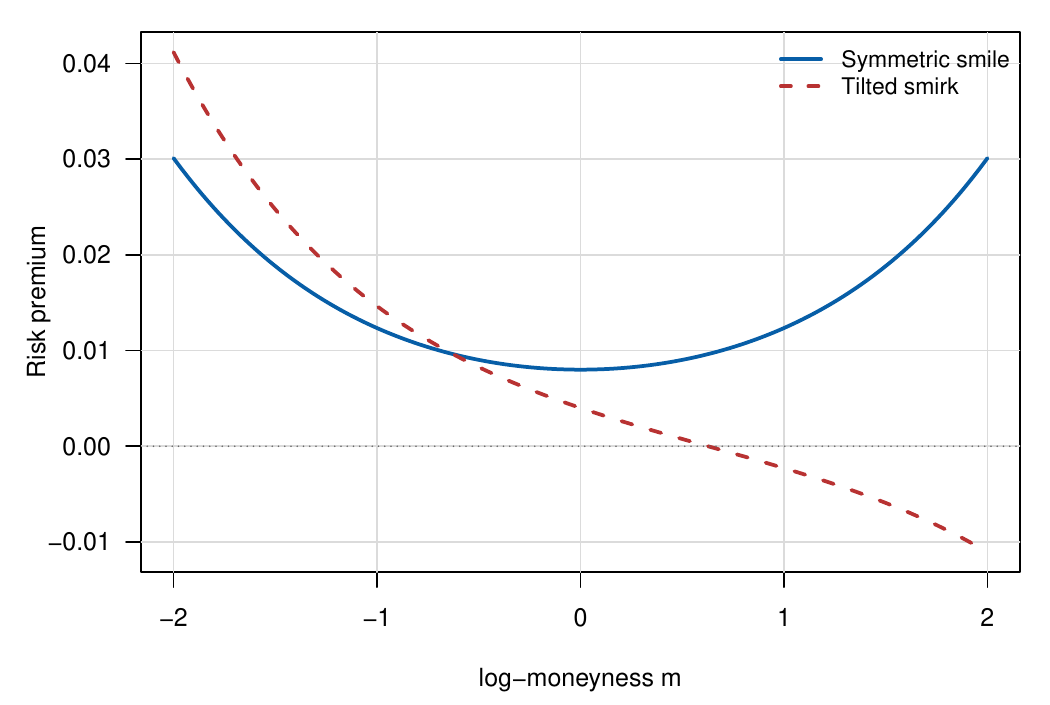}
\caption{Illustrative EBMM smile and smirk risk-premium shapes. The symmetric curve sets $A_{\mathrm{averse}}(\tau)=0.20$ and $A_{\mathrm{seek}}(\tau)=0$; the tilted curve sets $A_{\mathrm{averse}}(\tau)=0.10$ and $A_{\mathrm{seek}}(\tau)=-0.18$. The figure is schematic and is included to show the mechanism, not to claim empirical calibration.}
\label{fig:EBMM_smile_smirk}
\end{figure}

\subsection{Additional numerical diagnostics}\label{subsec:additional_numerical_diagnostics}
The EBMM construction has three moving parts: the Arrow--Pratt risk-attitude process, the stochastic-exponential density, and the aggregation of heterogeneous behavioral components. The following diagnostic plots are generated by the accompanying R script and are meant to clarify those mechanics. They are not calibrated to market data; they show how the theoretical restrictions behave under controlled parameter choices.

\Cref{fig:EBMM_arrow_pratt_paths} illustrates the recurrent Arrow--Pratt benchmark used in  \Cref{prop:APriskProcess}. The upper panel plots simulated Ornstein--Uhlenbeck risk-attitude states $R_t$, while the lower panel applies the positive transform $A_t=\exp(R_t)$. The figure makes the role of the transformation transparent for the locally risk-averse case: the recurrent benchmark can move above and below zero, while the transformed Arrow--Pratt coefficient used in the density remains positive. The signed empirical decomposition above separately allows negative local curvature as a risk-seeking component.

\begin{figure}[!htbp]
\centering
\includegraphics[width=0.72\textwidth]{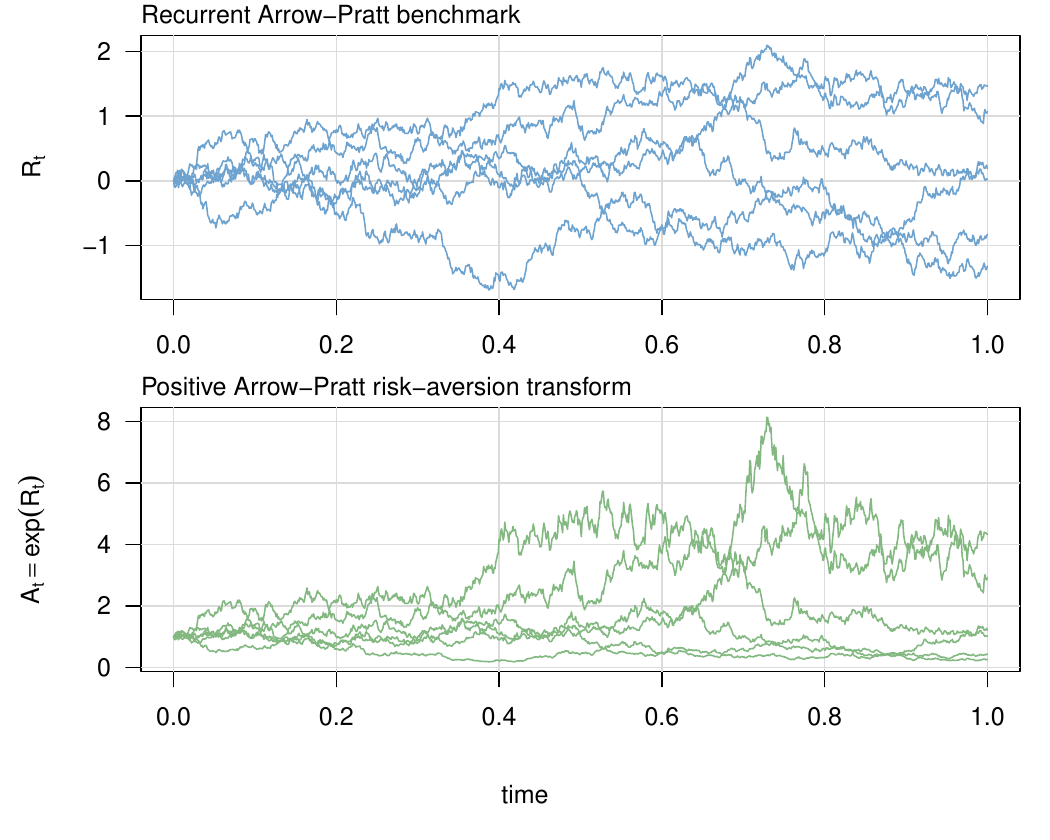}
\caption{Simulated recurrent Arrow--Pratt paths and positive risk-aversion transforms. The upper panel shows $R_t$; the lower panel shows $A_t=\exp(R_t)$.}
\label{fig:EBMM_arrow_pratt_paths}
\end{figure}

\Cref{fig:EBMM_density_paths} plots the corresponding behavioral martingale-density paths. Each path is generated from a stochastic exponential with kernel $-A_t\sigma_t$. The horizontal line at one marks the initial density. Dispersion around that line reflects state-price reweighting, not a change in total probability mass, because the density process is normalized as a martingale under the maintained admissibility condition.

\begin{figure}[!htbp]
\centering
\includegraphics[width=0.72\textwidth]{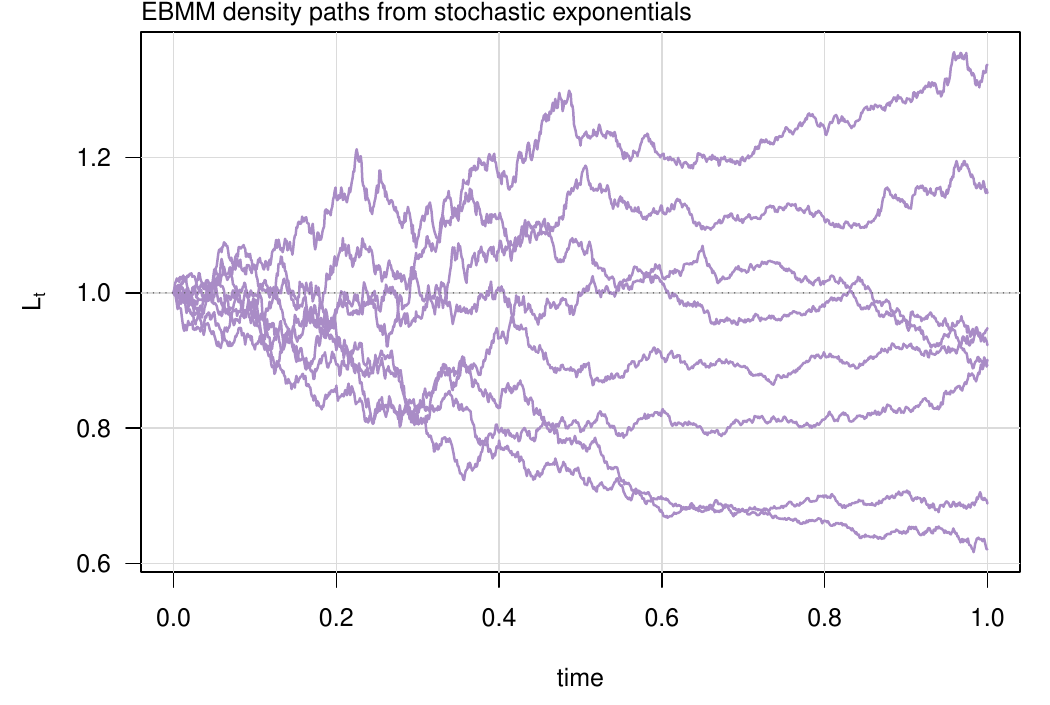}
\caption{Illustrative EBMM density paths generated by stochastic exponentials. The paths show how behavioral state-price weights vary over time while remaining normalized density processes.}
\label{fig:EBMM_density_paths}
\end{figure}

\Cref{fig:EBMM_weighted_mixture_density} illustrates the weighted EBMM construction in \Cref{thm:weightedEBMM}. The component densities correspond to low, medium, and high risk-attitude types. The black line is the normalized weighted mixture. The figure emphasizes why the weights must be economically interpretable and normalized: without normalization the mixture is not a probability density and therefore cannot define an equivalent martingale measure.

\begin{figure}[!htbp]
\centering
\includegraphics[width=0.72\textwidth]{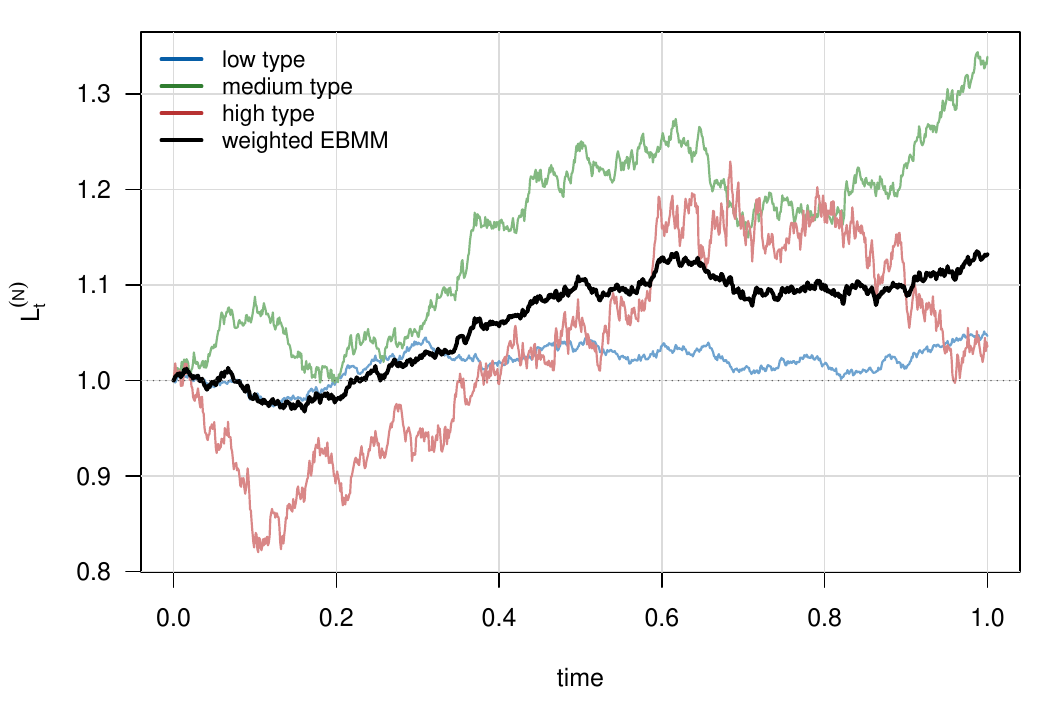}
\caption{Heterogeneous behavioral density components and the normalized weighted EBMM mixture. The component weights are interpreted as type shares or liquidity weights.}
\label{fig:EBMM_weighted_mixture_density}
\end{figure}

\tab Finally, \Cref{fig:EBMM_novikov_diagnostic} gives a simple diagnostic for the sufficient martingale condition. In the constant-kernel case $\theta=A_t\sigma_t$, Novikov's condition reduces to finiteness of an exponential square-integrability term. The plot shows how that term increases with the behavioral market-price-of-risk component. The condition is sufficient rather than necessary, but the diagnostic is useful because it identifies parameter regions where the stochastic-exponential density is safely admissible.

\begin{figure}[!htbp]
\centering
\includegraphics[width=0.70\textwidth]{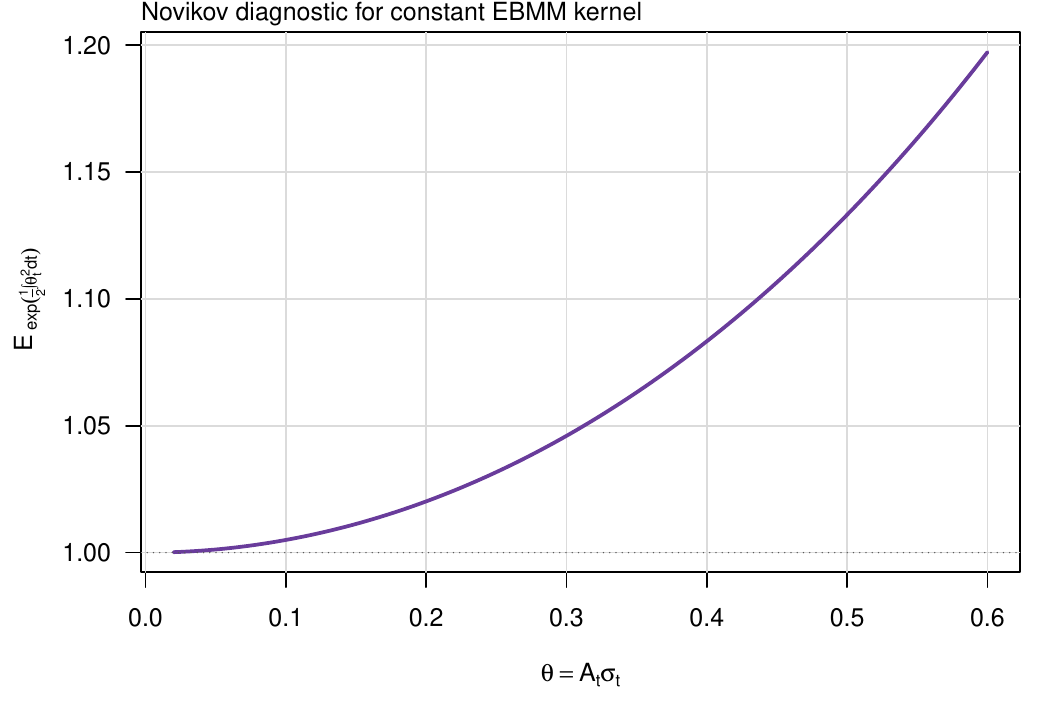}
\caption{Novikov diagnostic for a constant EBMM kernel $\theta=A_t\sigma_t$. The plotted term increases with behavioral risk sensitivity, illustrating why admissibility becomes more demanding as $A_t\sigma_t$ rises.}
\label{fig:EBMM_novikov_diagnostic}
\end{figure}

\subsection{Risk-management implications}\label{subsec:risk_management_implications}
The model yields three risk-management implications. First, when the product $A_t\sigma_t$ rises, the state-price density becomes more sensitive to adverse states, so hedging demand should increase even when physical drift estimates are unchanged. Second, option smiles and smirks can be interpreted as state-dependent risk-attitude surfaces rather than as volatility artifacts alone, provided the recovered density satisfies the EBMM admissibility restrictions. Third, heterogeneous-agent mixtures are economically meaningful only after raw pricing-kernel weights are normalized, as in \eqref{eq:normalizedWeights}; otherwise the mixture is not a probability density. These implications are testable with option panels by recovering state-price densities, estimating the physical density from realized returns, calibrating \eqref{eq:EBMM_calibration_objective}, and checking whether the implied marginal-utility density satisfies positivity, normalization, and stochastic-exponential admissibility.
\FloatBarrier
\section{\sc{Conclusion}}\label{sec:Conclusion}
This paper revises the Samuelson--Merton ``util-prob'' construction as a compensated marginal-utility density. The central point is that marginal utility alone does not define an equivalent martingale measure: positivity, integrability, normalization, and a stochastic-exponential martingale condition are required. The first main result gives sufficient conditions under which the compensated density is equivalent to the physical measure. The second shows how the Arrow--Pratt and prudence terms enter the compensator. The weighted EBMM result then shows how heterogeneous behavioral densities can be combined, provided the weights have an economic interpretation and are normalized.

\tab The risk-management section adds two implications. First, in a Cox--Ross setting the common state-price kernel cancels from the ratio of option and underlying risk premia; the EBMM interpretation identifies that kernel with risk attitude times volatility. Second, smiles and smirks can be represented by decomposing the Arrow--Pratt process into symmetric and asymmetric state-dependent risk-attitude components across moneyness. The hyperbolic specification is deliberately parsimonious and should be viewed as a tractable parametrization, not as the only possible smile model. The calibration protocol in \Cref{subsec:market_calibration} shows how the theory can be taken to data: recover the risk-neutral and physical densities, fit the EBMM density ratio, and test whether the fitted kernel satisfies the martingale admissibility restrictions. The present panel evidence suggests that the signed risk-seeking or smirk component carries more empirical variation than the symmetric risk-aversion component, and that it has preliminary small-sample predictive content for next-horizon realized variance. The empirical purpose of EBMM is therefore not to replace SVI, SABR, or other production volatility-surface models, but to ask whether the density ratio recovered from such models admits a compensated marginal-utility representation with interpretable signed Arrow--Pratt components. A full out-of-sample forecasting study with more dates, richer controls, liquidity filters, and production-quality arbitrage-free density recovery is a natural next step. Future work should implement this exercise on broader index, commodity, crypto, or single-name option panels, compare the EBMM decomposition with standard implied-volatility parametrizations, and test whether recovered marginal-utility densities satisfy the admissibility restrictions derived here.

\appendix
\inappendixtrue
\singlespacing
\section{Convex-Analytic Note on Mixtures of Pricing Measures}\label{app:topology}

\singlespacing
The mixture construction in \eqref{EBMMwgt1} can be read as a finite-dimensional convex representation: admissible pricing densities are combined with nonnegative weights and then normalized to preserve the martingale-density property. In compact convex settings, Krein--Milman and Choquet-type results provide a broader mathematical intuition for representing points of a convex set by mixtures over extreme points. That intuition is useful for thinking about heterogeneous-agent or regime mixtures of pricing kernels, but it is not needed for the paper's asset-pricing results. The operative object throughout the main text is the standard Radon--Nikodym density process: a positive, normalized martingale generated by a stochastic exponential. Thus the EBMM mixture is admissible because the component densities and the normalized weighted sum satisfy the stochastic-exponential and martingale restrictions, not because of a separate topological representation theorem.

\section{Additional Empirical Diagnostics}\label{app:empirical_diagnostics}

\tab This appendix reports diagnostic tables that support the main empirical exercise but are not needed for the core argument in \Cref{subsec:thetadata_calibration}. The main text emphasizes the sample summary, pooled benchmark comparison, structural Arrow--Pratt factor summary, finite-sample Novikov screen summary, and predictive regression. The full slice-level tables are retained here for replication transparency.

\begin{table}[!htbp]
\centering
\scriptsize
\caption{Out-of-sample benchmark errors by symbol and option type.}
\label{tab:EBMM_model_comparison_by_slice}
\resizebox{\textwidth}{!}{%
\begin{tabular}{lccccccc}
\hline
Group & Slices & Test $N$ & EBMM RMSE & Arb.-EBMM RMSE & Quadratic RMSE & Arb.-SVI RMSE & Wins \\
\hline
AAPL CALL & 6 & 348 & 0.0159 & 0.0159 & 0.0156 & 0.0112 & E*/E/Q/S=1/0/0/5 \\
AAPL PUT & 6 & 417 & 0.0299 & 0.0299 & 0.0298 & 0.0107 & E*/E/Q/S=0/0/0/6 \\
AMZN CALL & 6 & 369 & 0.0084 & 0.0084 & 0.0084 & 0.0094 & E*/E/Q/S=3/0/1/2 \\
AMZN PUT & 6 & 398 & 0.0263 & 0.0263 & 0.0264 & 0.0115 & E*/E/Q/S=0/0/0/6 \\
IWM CALL & 6 & 1022 & 0.0142 & 0.0142 & 0.0140 & 0.0074 & E*/E/Q/S=0/0/0/6 \\
IWM PUT & 6 & 1158 & 0.0294 & 0.0294 & 0.0293 & 0.0113 & E*/E/Q/S=0/0/0/6 \\
MSFT CALL & 6 & 553 & 0.0164 & 0.0164 & 0.0162 & 0.0119 & E*/E/Q/S=0/0/1/5 \\
MSFT PUT & 6 & 666 & 0.0298 & 0.0298 & 0.0297 & 0.0089 & E*/E/Q/S=0/0/0/6 \\
NVDA CALL & 6 & 883 & 0.0101 & 0.0101 & 0.0101 & 0.0116 & E*/E/Q/S=1/4/1/0 \\
NVDA PUT & 6 & 888 & 0.0200 & 0.0200 & 0.0201 & 0.0142 & E*/E/Q/S=0/0/0/6 \\
QQQ CALL & 6 & 919 & 0.0146 & 0.0146 & 0.0144 & 0.0117 & E*/E/Q/S=0/0/2/4 \\
QQQ PUT & 6 & 1119 & 0.0374 & 0.0374 & 0.0373 & 0.0112 & E*/E/Q/S=0/0/0/6 \\
SPY CALL & 6 & 1283 & 0.0181 & 0.0196 & 0.0179 & 0.0091 & E*/E/Q/S=0/0/0/6 \\
SPY PUT & 6 & 1550 & 0.0393 & 0.0393 & 0.0391 & 0.0117 & E*/E/Q/S=0/0/0/6 \\
TSLA CALL & 6 & 738 & 0.0104 & 0.0104 & 0.0104 & 0.0094 & E*/E/Q/S=0/0/0/6 \\
TSLA PUT & 6 & 738 & 0.0123 & 0.0123 & 0.0123 & 0.0121 & E*/E/Q/S=1/4/0/1 \\
\hline
\end{tabular}
}
\begin{minipage}{0.94\textwidth}
\footnotesize Notes: Each row refits the benchmark specifications on a fixed 70\% training split within each synchronized trade-date/maturity/right slice and reports averages across eligible slices for the indicated symbol/right group. The Wins column reports arbitrage-penalized EBMM/raw EBMM/quadratic/SVI slice wins. Overall win counts: Arb.-EBMM=6, EBMM=8, Quadratic=5, SVI=77. The arbitrage-penalized EBMM and SVI columns both use positivity and butterfly-density penalties on fitted total variance. The full slice-level table is saved in the replication CSV. These are forecasting-error diagnostics; a production comparison should still impose full static-arbitrage surface constraints.
\end{minipage}
\end{table}

\begin{table}[!htbp]
\centering
\caption{Call and put EBMM decompositions.}
\label{tab:EBMM_theta_by_right}
\begin{tabular}{lcc}
\hline
Parameter & Calls & Puts \\
\hline
$\widehat\beta_0$ & 0.2785 & 0.2938 \\
$\widehat\beta_+$ & 6.6937 & 7.4325 \\
$\widehat\beta_-$ & 0.1009 & -0.1754 \\
$R^2$ & 0.3410 & 0.5104 \\
$N$ & 20461 & 23181 \\
\hline
\end{tabular}
\end{table}

\begin{table}[!htbp]
\centering
\caption{Second-stage recovered Arrow--Pratt risk-attitude factors.}
\label{tab:EBMM_arrow_pratt_risk_attitudes}
\begin{tabular}{lcccc}
\hline
Sample & $\widehat A_{\mathrm{averse}}>0$ & $\widehat A_{\mathrm{seek}}<0$ & Density RMSE & $N$ \\
\hline
Pooled & 0.0002 & -4.0558 & 1.8046 & 12288 \\
AAPL/20250102/20250117 & 0.0000 & -15.7355 & 2.8791 & 256 \\
AAPL/20250102/20250221 & 0.0000 & -6.7925 & 1.9393 & 256 \\
AAPL/20250102/20250321 & 0.0000 & -10.4251 & 3.7955 & 256 \\
AAPL/20250203/20250221 & 0.0000 & -4.0168 & 0.7920 & 256 \\
AAPL/20250203/20250321 & 0.0000 & -4.4743 & 1.2006 & 256 \\
AAPL/20250303/20250321 & 0.0000 & -7.6608 & 1.5102 & 256 \\
AMZN/20250102/20250117 & 0.0000 & -10.6274 & 2.9111 & 256 \\
AMZN/20250102/20250221 & 0.0000 & -2.9319 & 1.5687 & 256 \\
AMZN/20250102/20250321 & 0.0000 & -1.1428 & 0.5160 & 256 \\
AMZN/20250203/20250221 & 14.6658 & -0.0000 & 0.3904 & 256 \\
AMZN/20250203/20250321 & 0.0000 & -1.5602 & 0.4317 & 256 \\
AMZN/20250303/20250321 & 0.0000 & -0.9619 & 0.4675 & 256 \\
IWM/20250102/20250117 & 0.0001 & -1.8410 & 0.1744 & 256 \\
IWM/20250102/20250221 & 47.9675 & -0.8267 & 0.2249 & 256 \\
IWM/20250102/20250321 & 78.2130 & -0.5566 & 0.3205 & 256 \\
IWM/20250203/20250221 & 0.0000 & -1.3945 & 0.2259 & 256 \\
IWM/20250203/20250321 & 5.5955 & -0.7943 & 0.3482 & 256 \\
IWM/20250303/20250321 & 115.1280 & -0.1252 & 0.3491 & 256 \\
MSFT/20250102/20250117 & 0.0000 & -16.6267 & 4.8147 & 256 \\
MSFT/20250102/20250221 & 0.0000 & -0.0000 & 0.5426 & 256 \\
MSFT/20250102/20250321 & 0.0000 & -9.9942 & 4.2902 & 256 \\
MSFT/20250203/20250221 & 0.0000 & -13.3685 & 3.9807 & 256 \\
MSFT/20250203/20250321 & 0.0000 & -0.0000 & 0.6809 & 256 \\
MSFT/20250303/20250321 & 0.0000 & -7.4815 & 2.0271 & 256 \\
NVDA/20250102/20250117 & 0.0000 & -2.3505 & 0.1955 & 256 \\
NVDA/20250102/20250221 & 0.0000 & -2.5110 & 0.1878 & 256 \\
NVDA/20250102/20250321 & 0.0000 & -2.1687 & 0.0852 & 256 \\
NVDA/20250203/20250221 & 24.4270 & -2.2114 & 0.0623 & 256 \\
NVDA/20250203/20250321 & 14.7557 & -2.1463 & 0.1604 & 256 \\
NVDA/20250303/20250321 & 35.4406 & -1.7055 & 0.0961 & 256 \\
QQQ/20250102/20250117 & 0.0000 & -9.7010 & 1.3545 & 256 \\
QQQ/20250102/20250221 & 0.0000 & -1.5229 & 0.7986 & 256 \\
QQQ/20250102/20250321 & 0.0000 & -7.8636 & 2.4815 & 256 \\
QQQ/20250203/20250221 & 0.0000 & -2.6659 & 0.4723 & 256 \\
QQQ/20250203/20250321 & 0.0000 & -13.4124 & 3.4444 & 256 \\
QQQ/20250303/20250321 & 0.0000 & -5.1526 & 0.8033 & 256 \\
SPY/20250102/20250117 & 0.0000 & -5.7608 & 0.3765 & 256 \\
SPY/20250102/20250221 & 0.0000 & -5.9081 & 0.7874 & 256 \\
SPY/20250102/20250321 & 16.8340 & -3.8705 & 0.3690 & 256 \\
SPY/20250203/20250221 & 0.0000 & -6.4564 & 0.5487 & 256 \\
SPY/20250203/20250321 & 0.0000 & -8.4089 & 1.2595 & 256 \\
SPY/20250303/20250321 & 0.0001 & -3.0957 & 0.4538 & 256 \\
TSLA/20250102/20250117 & 0.0000 & -0.6189 & 0.2190 & 256 \\
TSLA/20250102/20250221 & 0.0000 & -0.5909 & 0.0743 & 256 \\
TSLA/20250102/20250321 & 0.0000 & -0.7799 & 0.0667 & 256 \\
TSLA/20250203/20250221 & 0.0000 & -0.6165 & 0.1971 & 256 \\
TSLA/20250203/20250321 & 0.0000 & -0.6306 & 0.1981 & 256 \\
TSLA/20250303/20250321 & 1.7497 & -0.4110 & 0.1553 & 256 \\
\hline
\end{tabular}
\begin{minipage}{0.88\textwidth}
\footnotesize Notes: These estimates are obtained in a second-stage constrained minimum-distance fit of the normalized EBMM density kernel to the diagnostic $\log(q_\tau/p_\tau)$ ratio. The fitted log kernel is proportional to $A_{\mathrm{averse}}[\cosh(m)-1]+A_{\mathrm{seek}}\sinh(m)$ and is normalized against the physical density, so the intercept is absorbed by the density normalization. The theoretical restrictions are $A_{\mathrm{averse}}>0$ and $A_{\mathrm{seek}}<0$; $A_{\mathrm{averse}}$ is the risk-averse curvature component and $A_{\mathrm{seek}}$ is the signed risk-seeking/speculative component. The numerical optimizer imposes the interior bounds $A_{\mathrm{averse}}\geq\varepsilon$ and $A_{\mathrm{seek}}\leq-\varepsilon$ with $\varepsilon=10^{-8}$, so zero is excluded as the risk-neutral boundary.
\end{minipage}
\end{table}

\begin{table}[!htbp]
\centering
\caption{Raw-ratio second-stage Arrow--Pratt robustness fit.}
\label{tab:EBMM_arrow_pratt_raw_ratio_risk_attitudes}
\begin{tabular}{lcccc}
\hline
Sample & $\widehat A_{\mathrm{averse}}>0$ & $\widehat A_{\mathrm{seek}}<0$ & Raw-ratio RMSE & $N$ \\
\hline
Pooled & 4.00e-07 & -0.9595 & 0.7194 & 12288 \\
AAPL/20250102/20250117 & 5.18e-07 & -0.9295 & 0.6968 & 256 \\
AAPL/20250102/20250221 & 3.24e-08 & -1.1875 & 0.4315 & 256 \\
AAPL/20250102/20250321 & 7.89e-05 & -1.2846 & 0.5358 & 256 \\
AAPL/20250203/20250221 & 1.14e-06 & -1.2330 & 0.3458 & 256 \\
AAPL/20250203/20250321 & 1.76e-07 & -1.2245 & 0.3762 & 256 \\
AAPL/20250303/20250321 & 5.72e-08 & -1.1531 & 0.4262 & 256 \\
AMZN/20250102/20250117 & 1.00e-08 & -1.00e-08 & 0.5645 & 256 \\
AMZN/20250102/20250221 & 1.00e-08 & -1.00e-08 & 0.3342 & 256 \\
AMZN/20250102/20250321 & 7.85e-07 & -0.5212 & 0.4161 & 256 \\
AMZN/20250203/20250221 & 28.0859 & -3.21e-08 & 0.2888 & 256 \\
AMZN/20250203/20250321 & 5.84e-08 & -0.5902 & 0.3077 & 256 \\
AMZN/20250303/20250321 & 1.00e-08 & -1.00e-08 & 0.3397 & 256 \\
IWM/20250102/20250117 & 0.3630 & -1.5811 & 1.4565 & 256 \\
IWM/20250102/20250221 & 68.9838 & -9.89e-08 & 1.2083 & 256 \\
IWM/20250102/20250321 & 83.9593 & -1.54e-08 & 1.0489 & 256 \\
IWM/20250203/20250221 & 8.0404 & -0.9875 & 1.7805 & 256 \\
IWM/20250203/20250321 & 40.5779 & -1.0791 & 0.3194 & 256 \\
IWM/20250303/20250321 & 131.1926 & -0.7876 & 2.1790 & 256 \\
MSFT/20250102/20250117 & 1.00e-08 & -1.00e-08 & 0.6035 & 256 \\
MSFT/20250102/20250221 & 1.00e-08 & -1.00e-08 & 0.4677 & 256 \\
MSFT/20250102/20250321 & 9.08e-07 & -0.6729 & 0.5167 & 256 \\
MSFT/20250203/20250221 & 6.88e-06 & -0.3779 & 0.5241 & 256 \\
MSFT/20250203/20250321 & 1.00e-08 & -1.00e-08 & 0.5724 & 256 \\
MSFT/20250303/20250321 & 3.40e-07 & -0.1867 & 0.4539 & 256 \\
NVDA/20250102/20250117 & 1.19e-05 & -1.9434 & 0.2305 & 256 \\
NVDA/20250102/20250221 & 1.84e-07 & -1.8325 & 0.1385 & 256 \\
NVDA/20250102/20250321 & 0.7700 & -2.0461 & 0.0957 & 256 \\
NVDA/20250203/20250221 & 23.8908 & -2.3525 & 0.1403 & 256 \\
NVDA/20250203/20250321 & 12.1720 & -2.1381 & 0.2132 & 256 \\
NVDA/20250303/20250321 & 34.4335 & -2.0005 & 0.1942 & 256 \\
QQQ/20250102/20250117 & 8.40e-07 & -1.8509 & 0.6054 & 256 \\
QQQ/20250102/20250221 & 8.84e-07 & -0.4509 & 0.5328 & 256 \\
QQQ/20250102/20250321 & 2.64e-06 & -1.5602 & 0.3643 & 256 \\
QQQ/20250203/20250221 & 5.56e-07 & -1.2108 & 0.5037 & 256 \\
QQQ/20250203/20250321 & 8.79e-07 & -1.7213 & 0.4607 & 256 \\
QQQ/20250303/20250321 & 5.11e-07 & -1.9675 & 0.3153 & 256 \\
SPY/20250102/20250117 & 6.25e-06 & -2.9678 & 0.9125 & 256 \\
SPY/20250102/20250221 & 4.86e-07 & -2.6550 & 0.3474 & 256 \\
SPY/20250102/20250321 & 42.3000 & -3.05e-08 & 0.7329 & 256 \\
SPY/20250203/20250221 & 5.74e-04 & -3.1682 & 0.9478 & 256 \\
SPY/20250203/20250321 & 4.11e-05 & -2.6860 & 0.3591 & 256 \\
SPY/20250303/20250321 & 81.4091 & -2.3823 & 1.4268 & 256 \\
TSLA/20250102/20250117 & 3.00e-08 & -0.3347 & 0.1880 & 256 \\
TSLA/20250102/20250221 & 1.79e-08 & -0.5103 & 0.0725 & 256 \\
TSLA/20250102/20250321 & 1.54e-08 & -0.6796 & 0.0667 & 256 \\
TSLA/20250203/20250221 & 7.50e-08 & -0.4161 & 0.1768 & 256 \\
TSLA/20250203/20250321 & 4.60e-07 & -0.4127 & 0.1842 & 256 \\
TSLA/20250303/20250321 & 3.0106 & -0.3564 & 0.1602 & 256 \\
\hline
\end{tabular}
\begin{minipage}{0.88\textwidth}
\footnotesize Notes: This robustness check fits the raw diagnostic ratio $q_\tau/p_\tau$ directly, rather than its logarithm. Because the ratio can become numerically singular when the physical density is very small, the objective uses physical-density integration weights, a 99.5\% tail cap on the empirical ratio within each synchronized slice, and a Huber loss on level residuals. The raw-ratio specification is intentionally severe and is sensitive to tail singularities when $\widehat p_\tau$ is small. The fact that the sign restriction on $A_{\mathrm{seek}}$ survives is informative, but the near-boundary estimates of $A_{\mathrm{averse}}$ in several short-horizon slices show that raw-level density ratios should not be used as the primary structural estimator. Scientific notation flags estimates that are very close to, but still inside, the risk-neutral boundary.
\end{minipage}
\end{table}

\begin{table}[!htbp]
\centering
\caption{Finite-sample Novikov admissibility screen for the option-panel calibration.}
\label{tab:EBMM_novikov_sample_diagnostic}
\begin{tabular}{lccccc}
\hline
Slice & $\tau$ & Mean Novikov term & Max Novikov term & Max $|\widehat\lambda|$ & Screen \\
\hline
AAPL/20250102/20250117 & 0.041 & 1.0682 & 1.7352 & 5.1808 & Yes \\
AAPL/20250102/20250221 & 0.137 & 1.0239 & 1.1736 & 1.5294 & Yes \\
AAPL/20250102/20250321 & 0.214 & 1.0579 & 1.4858 & 1.9257 & Yes \\
AAPL/20250203/20250221 & 0.049 & 1.0032 & 1.0353 & 1.1871 & Yes \\
AAPL/20250203/20250321 & 0.126 & 1.0090 & 1.0635 & 0.9890 & Yes \\
AAPL/20250303/20250321 & 0.049 & 1.0175 & 1.1489 & 2.3734 & Yes \\
AMZN/20250102/20250117 & 0.041 & 1.0250 & 1.2569 & 3.3372 & Yes \\
AMZN/20250102/20250221 & 0.137 & 1.0047 & 1.0275 & 0.6295 & Yes \\
AMZN/20250102/20250321 & 0.214 & 1.0009 & 1.0052 & 0.2199 & Yes \\
AMZN/20250203/20250221 & 0.049 & 1.0012 & 1.0136 & 0.7402 & Yes \\
AMZN/20250203/20250321 & 0.126 & 1.0014 & 1.0078 & 0.3504 & Yes \\
AMZN/20250303/20250321 & 0.049 & 1.0003 & 1.0023 & 0.3083 & Yes \\
IWM/20250102/20250117 & 0.041 & 1.0006 & 1.0065 & 0.5617 & Yes \\
IWM/20250102/20250221 & 0.137 & 1.0150 & 1.1695 & 1.5124 & Yes \\
IWM/20250102/20250321 & 0.214 & 1.0245 & 1.7882 & 2.3330 & Yes \\
IWM/20250203/20250221 & 0.049 & 1.0004 & 1.0041 & 0.4080 & Yes \\
IWM/20250203/20250321 & 0.126 & 1.0002 & 1.0047 & 0.2720 & Yes \\
IWM/20250303/20250321 & 0.049 & 1.0344 & 2.5082 & 6.1089 & Yes \\
MSFT/20250102/20250117 & 0.041 & 1.0844 & 1.8477 & 5.4680 & Yes \\
MSFT/20250102/20250221 & 0.137 & 1.0000 & 1.0000 & 0.0000 & Yes \\
MSFT/20250102/20250321 & 0.214 & 1.0652 & 1.5040 & 1.9551 & Yes \\
MSFT/20250203/20250221 & 0.049 & 1.0418 & 1.4721 & 3.9616 & Yes \\
MSFT/20250203/20250321 & 0.126 & 1.0000 & 1.0000 & 0.0000 & Yes \\
MSFT/20250303/20250321 & 0.049 & 1.0161 & 1.1392 & 2.3000 & Yes \\
NVDA/20250102/20250117 & 0.041 & 1.0021 & 1.0119 & 0.7591 & Yes \\
NVDA/20250102/20250221 & 0.137 & 1.0050 & 1.0224 & 0.5693 & Yes \\
NVDA/20250102/20250321 & 0.214 & 1.0063 & 1.0233 & 0.4644 & Yes \\
NVDA/20250203/20250221 & 0.049 & 1.0075 & 1.0992 & 1.9589 & Yes \\
NVDA/20250203/20250321 & 0.126 & 1.0099 & 1.1122 & 1.2995 & Yes \\
NVDA/20250303/20250321 & 0.049 & 1.0151 & 1.1848 & 2.6235 & Yes \\
QQQ/20250102/20250117 & 0.041 & 1.0158 & 1.2194 & 3.1084 & Yes \\
QQQ/20250102/20250221 & 0.137 & 1.0006 & 1.0048 & 0.2652 & Yes \\
QQQ/20250102/20250321 & 0.214 & 1.0338 & 1.2922 & 1.5495 & Yes \\
QQQ/20250203/20250221 & 0.049 & 1.0006 & 1.0097 & 0.6259 & Yes \\
QQQ/20250203/20250321 & 0.126 & 1.0859 & 1.8367 & 3.1072 & Yes \\
QQQ/20250303/20250321 & 0.049 & 1.0054 & 1.0717 & 1.6765 & Yes \\
SPY/20250102/20250117 & 0.041 & 1.0053 & 1.0764 & 1.8938 & Yes \\
SPY/20250102/20250221 & 0.137 & 1.0062 & 1.0868 & 1.1028 & Yes \\
SPY/20250102/20250321 & 0.214 & 1.0093 & 1.1191 & 1.0264 & Yes \\
SPY/20250203/20250221 & 0.049 & 1.0043 & 1.0527 & 1.4441 & Yes \\
SPY/20250203/20250321 & 0.126 & 1.0123 & 1.1850 & 1.6417 & Yes \\
SPY/20250303/20250321 & 0.049 & 1.0012 & 1.0196 & 0.8872 & Yes \\
TSLA/20250102/20250117 & 0.041 & 1.0001 & 1.0009 & 0.2067 & Yes \\
TSLA/20250102/20250221 & 0.137 & 1.0005 & 1.0015 & 0.1480 & Yes \\
TSLA/20250102/20250321 & 0.214 & 1.0011 & 1.0037 & 0.1850 & Yes \\
TSLA/20250203/20250221 & 0.049 & 1.0002 & 1.0009 & 0.1901 & Yes \\
TSLA/20250203/20250321 & 0.126 & 1.0004 & 1.0016 & 0.1601 & Yes \\
TSLA/20250303/20250321 & 0.049 & 1.0001 & 1.0013 & 0.2273 & Yes \\
\hline
\end{tabular}
\begin{minipage}{0.90\textwidth}
\footnotesize Notes: The screen reports $N_j^{-1}\sum_i\exp\{0.5\tau_j\widehat\lambda_{ij}^2\}$ with $\widehat\lambda_{ij}=\widehat A_\theta(m_{ij},\tau_j)\widehat\sigma_{ij}$. All reported slice-level values are finite when the screen column is Yes, consistent with the maintained admissibility assumption. This is a finite-sample consistency screen, not a proof of the continuous-time Novikov condition.
\end{minipage}
\end{table}

\begin{table}[!htbp]
\centering
\caption{Approximate recovered $\log(q_\tau/p_\tau)$ diagnostic.}
\label{tab:EBMM_qp_ratio_summary}
\begin{tabular}{lllcccc}
\hline
Symbol & Trade date & Expiration & $\tau$ & 5\% & Median & 95\% \\
\hline
AAPL & 20250102 & 20250117 & 0.041 & -25.404 & -1.086 & 1.541 \\
AAPL & 20250102 & 20250221 & 0.137 & -24.990 & -0.272 & 0.437 \\
AAPL & 20250102 & 20250321 & 0.214 & -26.134 & -0.139 & 0.580 \\
AAPL & 20250203 & 20250221 & 0.049 & -22.948 & -0.278 & 1.113 \\
AAPL & 20250203 & 20250321 & 0.126 & -3.781 & -0.225 & 0.388 \\
AAPL & 20250303 & 20250321 & 0.049 & -24.078 & -0.370 & 1.844 \\
AMZN & 20250102 & 20250117 & 0.041 & -25.418 & -1.320 & 0.512 \\
AMZN & 20250102 & 20250221 & 0.137 & -2.574 & -0.349 & 0.323 \\
AMZN & 20250102 & 20250321 & 0.214 & -1.517 & -0.543 & 0.443 \\
AMZN & 20250203 & 20250221 & 0.049 & -0.819 & -0.255 & 3.796 \\
AMZN & 20250203 & 20250321 & 0.126 & -2.173 & -0.285 & 0.332 \\
AMZN & 20250303 & 20250321 & 0.049 & -1.394 & -0.463 & 0.314 \\
IWM & 20250102 & 20250117 & 0.041 & -0.384 & 0.336 & 23.509 \\
IWM & 20250102 & 20250221 & 0.137 & -0.311 & 1.446 & 8.070 \\
IWM & 20250102 & 20250321 & 0.214 & -0.610 & 0.501 & 6.417 \\
IWM & 20250203 & 20250221 & 0.049 & -14.473 & 0.084 & 23.616 \\
IWM & 20250203 & 20250321 & 0.126 & -21.088 & 0.217 & 6.905 \\
IWM & 20250303 & 20250321 & 0.049 & -21.356 & 0.292 & 23.903 \\
MSFT & 20250102 & 20250117 & 0.041 & -26.585 & -1.800 & 0.527 \\
MSFT & 20250102 & 20250221 & 0.137 & -1.201 & -0.518 & 1.640 \\
MSFT & 20250102 & 20250321 & 0.214 & -26.485 & -0.491 & 0.510 \\
MSFT & 20250203 & 20250221 & 0.049 & -26.163 & -1.586 & 0.457 \\
MSFT & 20250203 & 20250321 & 0.126 & -2.018 & -0.588 & 1.997 \\
MSFT & 20250303 & 20250321 & 0.049 & -24.872 & -1.008 & 0.390 \\
NVDA & 20250102 & 20250117 & 0.041 & -0.889 & 0.171 & 2.570 \\
NVDA & 20250102 & 20250221 & 0.137 & -1.358 & 0.074 & 0.719 \\
NVDA & 20250102 & 20250321 & 0.214 & -0.981 & 0.033 & 1.115 \\
NVDA & 20250203 & 20250221 & 0.049 & -0.204 & 0.324 & 3.492 \\
NVDA & 20250203 & 20250321 & 0.126 & -0.385 & -0.087 & 1.559 \\
NVDA & 20250303 & 20250321 & 0.049 & -0.286 & 0.409 & 3.885 \\
QQQ & 20250102 & 20250117 & 0.041 & -23.224 & -0.006 & 14.427 \\
QQQ & 20250102 & 20250221 & 0.137 & -23.091 & 0.047 & 2.417 \\
QQQ & 20250102 & 20250321 & 0.214 & -25.303 & -0.139 & 1.013 \\
QQQ & 20250203 & 20250221 & 0.049 & -13.945 & 0.373 & 6.012 \\
QQQ & 20250203 & 20250321 & 0.126 & -25.854 & -0.307 & 1.756 \\
QQQ & 20250303 & 20250321 & 0.049 & -23.085 & -0.063 & 6.529 \\
SPY & 20250102 & 20250117 & 0.041 & -0.948 & -0.006 & 23.013 \\
SPY & 20250102 & 20250221 & 0.137 & -22.851 & -0.114 & 4.426 \\
SPY & 20250102 & 20250321 & 0.214 & -0.770 & 0.911 & 4.665 \\
SPY & 20250203 & 20250221 & 0.049 & -21.777 & -0.008 & 23.050 \\
SPY & 20250203 & 20250321 & 0.126 & -23.719 & -0.082 & 5.041 \\
SPY & 20250303 & 20250321 & 0.049 & -21.627 & 0.017 & 23.422 \\
TSLA & 20250102 & 20250117 & 0.041 & -0.930 & -0.120 & 0.348 \\
TSLA & 20250102 & 20250221 & 0.137 & -0.321 & 0.050 & 0.124 \\
TSLA & 20250102 & 20250321 & 0.214 & -0.444 & 0.060 & 0.182 \\
TSLA & 20250203 & 20250221 & 0.049 & -0.738 & -0.195 & 0.182 \\
TSLA & 20250203 & 20250321 & 0.126 & -0.656 & -0.123 & 0.203 \\
TSLA & 20250303 & 20250321 & 0.049 & -0.594 & -0.031 & 1.112 \\
\hline
\end{tabular}
\begin{minipage}{0.88\textwidth}
\footnotesize Notes: The ratio is a diagnostic proxy. The risk-neutral density is recovered by applying the Breeden--Litzenberger second-derivative relation to the arbitrage-penalized SVI Black--Scholes call curve on each synchronized slice. The physical density uses FMPCloud adjusted daily stock history when configured, then ThetaData stock history if available, and otherwise the underlying-price series embedded in the option-implied-volatility endpoint. The physical density is estimated as a 50/50 blend of overlapping-horizon kernel density estimates and an EWMA filtered-innovation bootstrap. This is an upgraded slice-level recovery, but not a full production surface with global calendar-spread constraints. Large upper-tail log ratios, especially at short maturities, flag tail singularities where $\widehat p_\tau$ is small and should be read as a warning about the diagnostic rather than as structural evidence.
\end{minipage}
\end{table}

\begin{table}[!htbp]
\centering
\caption{Density-ratio tail-instability diagnostics.}
\label{tab:EBMM_density_tail_diagnostics}
\begin{tabular}{lccccc}
\hline
Slice & $\tau$ & 99\% $\log(q/p)$ & Max $q/p$ & Min $\widehat p$ & Flag \\
\hline
AAPL/20250102/20250117 & 0.041 & 2.015 & 1.96e+01 & 4.83e-04 & No \\
AAPL/20250102/20250221 & 0.137 & 0.495 & 1.65e+00 & 2.57e-02 & No \\
AAPL/20250102/20250321 & 0.214 & 0.610 & 1.84e+00 & 5.31e-02 & No \\
AAPL/20250203/20250221 & 0.049 & 1.749 & 6.63e+00 & 5.63e-04 & No \\
AAPL/20250203/20250321 & 0.126 & 0.430 & 1.54e+00 & 2.89e-02 & No \\
AAPL/20250303/20250321 & 0.049 & 2.507 & 1.41e+01 & 5.63e-04 & No \\
AMZN/20250102/20250117 & 0.041 & 0.673 & 1.98e+00 & 1.35e-02 & No \\
AMZN/20250102/20250221 & 0.137 & 0.357 & 1.43e+00 & 1.53e-01 & No \\
AMZN/20250102/20250321 & 0.214 & 0.468 & 1.60e+00 & 2.22e-01 & No \\
AMZN/20250203/20250221 & 0.049 & 4.961 & 1.89e+02 & 2.43e-02 & No \\
AMZN/20250203/20250321 & 0.126 & 0.355 & 1.43e+00 & 1.42e-01 & No \\
AMZN/20250303/20250321 & 0.049 & 0.364 & 1.44e+00 & 2.43e-02 & No \\
IWM/20250102/20250117 & 0.041 & 23.919 & 2.75e+10 & 1.00e-12 & Yes \\
IWM/20250102/20250221 & 0.137 & 12.330 & 8.43e+05 & 8.67e-07 & No \\
IWM/20250102/20250321 & 0.214 & 8.270 & 4.09e+03 & 3.41e-04 & No \\
IWM/20250203/20250221 & 0.049 & 23.926 & 2.68e+10 & 1.00e-12 & Yes \\
IWM/20250203/20250321 & 0.126 & 9.964 & 1.09e+05 & 2.09e-07 & No \\
IWM/20250303/20250321 & 0.049 & 24.434 & 4.69e+10 & 1.00e-12 & Yes \\
MSFT/20250102/20250117 & 0.041 & 0.840 & 2.37e+00 & 1.25e-02 & No \\
MSFT/20250102/20250221 & 0.137 & 2.363 & 1.23e+01 & 1.01e-01 & No \\
MSFT/20250102/20250321 & 0.214 & 0.531 & 1.70e+00 & 1.55e-01 & No \\
MSFT/20250203/20250221 & 0.049 & 0.543 & 1.73e+00 & 2.22e-02 & No \\
MSFT/20250203/20250321 & 0.126 & 2.375 & 1.15e+01 & 9.96e-02 & No \\
MSFT/20250303/20250321 & 0.049 & 0.452 & 1.58e+00 & 2.22e-02 & No \\
NVDA/20250102/20250117 & 0.041 & 2.779 & 1.75e+01 & 1.95e-03 & No \\
NVDA/20250102/20250221 & 0.137 & 0.749 & 2.12e+00 & 8.71e-02 & No \\
NVDA/20250102/20250321 & 0.214 & 1.368 & 4.21e+00 & 9.33e-02 & No \\
NVDA/20250203/20250221 & 0.049 & 5.393 & 3.01e+02 & 2.50e-04 & No \\
NVDA/20250203/20250321 & 0.126 & 1.725 & 6.05e+00 & 6.71e-02 & No \\
NVDA/20250303/20250321 & 0.049 & 5.823 & 4.62e+02 & 2.50e-04 & No \\
QQQ/20250102/20250117 & 0.041 & 22.566 & 6.74e+09 & 1.00e-12 & Yes \\
QQQ/20250102/20250221 & 0.137 & 2.598 & 1.37e+01 & 4.34e-03 & No \\
QQQ/20250102/20250321 & 0.214 & 1.076 & 3.36e+00 & 1.61e-02 & No \\
QQQ/20250203/20250221 & 0.049 & 16.316 & 4.31e+08 & 1.00e-12 & Yes \\
QQQ/20250203/20250321 & 0.126 & 2.011 & 7.64e+00 & 3.96e-03 & No \\
QQQ/20250303/20250321 & 0.049 & 16.847 & 7.35e+08 & 1.00e-12 & Yes \\
SPY/20250102/20250117 & 0.041 & 23.131 & 1.15e+10 & 1.00e-12 & Yes \\
SPY/20250102/20250221 & 0.137 & 11.170 & 7.18e+05 & 1.00e-12 & Yes \\
SPY/20250102/20250321 & 0.214 & 6.396 & 1.28e+03 & 1.24e-04 & No \\
SPY/20250203/20250221 & 0.049 & 23.308 & 1.42e+10 & 1.00e-12 & Yes \\
SPY/20250203/20250321 & 0.126 & 10.672 & 5.22e+05 & 1.00e-12 & Yes \\
SPY/20250303/20250321 & 0.049 & 23.746 & 2.25e+10 & 1.00e-12 & Yes \\
TSLA/20250102/20250117 & 0.041 & 0.646 & 2.12e+00 & 1.47e-02 & No \\
TSLA/20250102/20250221 & 0.137 & 0.145 & 1.19e+00 & 2.62e-01 & No \\
TSLA/20250102/20250321 & 0.214 & 0.185 & 1.20e+00 & 4.12e-01 & No \\
TSLA/20250203/20250221 & 0.049 & 0.190 & 1.21e+00 & 3.23e-02 & No \\
TSLA/20250203/20250321 & 0.126 & 0.209 & 1.23e+00 & 2.60e-01 & No \\
TSLA/20250303/20250321 & 0.049 & 1.267 & 3.59e+00 & 3.23e-02 & No \\
\hline
\end{tabular}
\begin{minipage}{0.90\textwidth}
\footnotesize Notes: The flag is triggered when the 99\% log-ratio exceeds $\log(10^6)$ or the estimated physical density falls below $10^{-8}$. These diagnostics are used to identify tail singularities before interpreting raw $q/p$ fits.
\end{minipage}
\end{table}

\singlespacing
\bibliographystyle{chicago}      
\section*{}
\addcontentsline{toc}{section}{References} 
\bibliography{EBMMBIB}           
\end{document}